\documentclass[a4paper,reqno,twoside,11pt]{amsart}%
\usepackage{amssymb}
\usepackage{amsfonts}
\usepackage{amsmath}
\usepackage{graphicx}
\usepackage[a4paper,asymmetric]{geometry}
\usepackage{times}
\usepackage{subfig}%
\providecommand{\U}[1]{\protect\rule{.1in}{.1in}}
\newtheorem{theorem}{Theorem}[section]
\newtheorem{lemma}[theorem]{Lemma}
\theoremstyle{definition}

\newtheorem{example}[theorem]{Example}

\theoremstyle{remark}
\newtheorem{remark}[theorem]{Remark}
\numberwithin{equation}{section}
\theoremstyle{plain}

\newtheorem{proposition}{Proposition}

\newtheorem{assumption}{Assumption}
\allowdisplaybreaks
\copyrightinfo{2001}{enter name of copyright holder}
\ifx\pdfoutput\relax\let\pdfoutput=\undefined\fi
\newcount\msipdfoutput
\ifx\pdfoutput\undefined\else
\ifcase\pdfoutput\else
\ifx\paperwidth\undefined\else
\ifdim\paperheight=0pt\relax\else\pdfpageheight\paperheight\fi
\ifdim\paperwidth=0pt\relax\else\pdfpagewidth\paperwidth\fi
\fi\fi\fi
\begin{document}
\title[A change of measure in the BNS model for commodities]{A change of measure preserving the affine structure in the BNS model for
commodity markets}
\date{\today}
\author[Benth]{Fred Espen Benth}
\address[Fred Espen Benth]{\\
Centre of Mathematics for Applications \\
University of Oslo\\
P.O. Box 1053, Blindern\\
N--0316 Oslo, Norway}
\email[]{fredb\@@math.uio.no}
\urladdr{http://folk.uio.no/fredb/}
\author[Ortiz-Latorre]{Salvador Ortiz-Latorre}
\address[Salvador Ortiz-Latorre]{\\
Centre of Mathematics for Applications \\
University of Oslo\\
P.O. Box 1053, Blindern\\
N--0316 Oslo, Norway}
\email[]{salvador.ortiz-latorre\@@cma.uio.no}
\thanks{We are grateful for the financial support from the project "Energy Markets:
Modeling, Optimization and Simulation (EMMOS)", funded by the Norwegian
Research Council under grant Evita/205328.}
\maketitle

\begin{abstract}
For a commodity spot price dynamics given by an Ornstein-Uhlenbeck process
with Barndorff-Nielsen and Shephard stochastic volatility, we price forwards
using a class of pricing measures that simultaneously allow for change of
level and speed in the mean reversion of both the price and the volatility.
The risk premium is derived in the case of arithmetic and geometric spot price
processes, and it is demonstrated that we can provide flexible shapes that is
typically observed in energy markets. In particular, our pricing measure
preserves the affine model structure and decomposes into a price and
volatility risk premium, and in the geometric spot price model we need to
resort to a detailed analysis of a system of Riccati equations, for which we
show existence and uniqueness of solution and asymptotic properties that
explains the possible risk premium profiles. Among the typical shapes, the
risk premium allows for a stochastic change of sign, and can attain positive
values in the short end of the forward market and negative in the long end.

\end{abstract}

\section{Introduction}

Benth and Ortiz-Latorre~\cite{BeO-L13} analysed a structure preserving class
of pricing measures for Ornstein-Uhlenbeck (OU) processes with applications to
forward pricing in commodity markets. In particular, they considered
multi-factor OU models driven by L\'evy processes having positive jumps
(so-called subordinators) or Brownian motions for the spot price dynamics and
analysed the risk premium when the level and speed of mean reversion in these
factor processes were changed.

In this paper we continue this study for OU processes driven by Brownian
motion, but with a stochastic volatility perturbing the driving noise. The
stochastic volatility process is modelled again as an OU process, but driven
by a subordinator. This class of stochastic volatility models were first
introduced by Barndorff-Nielsen and Shephard~\cite{BNS} for equity prices, and
later analysed by Benth~\cite{Be11} in commodity markets. Indeed, the present
paper is considering a class of pricing measures preserving the affine
structure of the spot price model analysed in Benth~\cite{Be11}.

Our spot price dynamics is a generalization of the Schwartz model (see
Schwartz~\cite{Schwartz}) to account for stochastic volatility. The Schwartz
model have been applied to many different commodity markets, including oil
(see Schwartz~\cite{Schwartz}), power (see Lucia and Schwartz~\cite{LS}),
weather (see Benth and \v{S}altyt\.{e} Benth~\cite{BSB-book}) and freight (see
Benth, Koekebakker and Taib~\cite{BKT}). Like Lucia and Schwartz~\cite{LS}, we
analyse both geometric and arithmetic models for the spot price evolution.
There exists many extensions of the model, typically allowing for more factors
in the spot price dynamics, as well as modelling the convenience yield and
interest rates (see Eydeland and Wolyniec~\cite{EW} and Geman~\cite{Geman} for
more on such models). In Benth~\cite{Be11}, the Schwartz model with stochastic
volatility has been applied to model empirically UK gas prices. Also other
stochastic volatility models like the Heston have been suggested in the
context of commodity markets (see Eydeland and Wolyniec~\cite{EW} and
Geman~\cite{Geman} for a discussion and further references).

The class of pricing measures we study here allows for a simultaneous change
of speed and level of mean reversion for both the (logarithmic) spot price and
the stochastic volatility process. The mean reversion level can be flexibly
shifted up or down, while the speed of mean reversion can be slowed down. It
significantly extends the Esscher transform, which only allows for changes in
the level of mean reversion. Indeed, it decomposes the risk premium into a
price and volatility premium. It has been studied empirically in some
commodity markets for multi-factor models in Benth, Cartea and
Pedraz~\cite{BCP}. As we show, the class of pricing measures preserves the
affine structure of the model, but leads to a rather complex stochastic driver
for the stochastic volatility. For the arithmetic spot model we can derive
analytic forward prices and risk premium curves. On the other hand, the
geometric model is far more complex, but the affine structure can be exploited
to reduce the forward pricing to solving a system of Riccati equations by
resorting to the theory of Kallsen and Muhle-Karbe~\cite{KaMu-ka10}. The
forward price becomes a function of both the spot and the volatility, and has
a deterministic asymptotic dynamics when we are far from maturity.

By careful analysis of the associated system of Riccati equations, we can
study the implied risk premium of our class of measure change as a function of
its parameters. The risk premium is defined as the difference between the
forward price and the predicted spot price at maturity, and is a notion of
great importance in commodity markets since it measures the price for entering
a forward hedge position in the commodity (see e.g. Geman~\cite{Geman} for
more on this). In particular, under rather mild assumptions on the parameters,
we can show that the risk premium may change sign stochastically, and may be
positive for short times to maturity and negative when maturity is farther out
in time. This is a profile of the risk premium that one may expect in power
markets based on both economical and empirical findings. Geman and
Vasicek~\cite{GV} argue that retailers in the power market may induce a
hedging pressure by entering long positions in forwards to protect themselves
against sudden price increases (spikes). This may lead to positive risk
premia, whereas producers induce a negative premium in the long end of the
forward curve since they hedge by selling their production. This economic
argument for a positive premium in the short end is backed up by empirical
evidence from the German power market found in Benth, Cartea and
Kiesel~\cite{BCK}. In the geometric model, we show that the sign of the risk
premium depends explicitly on the current level of the logarithmic spot price.

We recover the Esscher transform in a special case of our pricing measure. The
Esscher transform is a popular tool for introducing a pricing measure in
commodity markets, or, equivalently, to model the risk premium. For constant
market prices of risk, which are defined as the shift in level of mean
reversion, we preserve the affine structure of the model as well as the
L\'{e}vy property of the driving noises of the two OU processes that we
consider (indeed, the spot price dynamics is driven by a Brownian motion). We
find such pricing measures in for example Lucia and Schwartz~\cite{LS}, Kolos
and Ronn~\cite{KR} and Schwartz and Smith~\cite{SS}. We refer the reader to
Benth, \v{S}altyt\.{e} Benth and Koekebakker~\cite{BeSa-BeKo08} for a thorough
discussion and references to the application of Esscher transform in power and
related markets. We note that the Esscher transform was first introduced and
applied to insurance as a tool to model the premium charged for covering a
given risk exposure and later adopted in pricing in incomplete financial
markets (see Gerber and Shiu~\cite{GS}). In many ways, in markets where the
underlying commodity is not storable (that is, cannot be traded in a
portfolio), the pricing of forwards and futures can be viewed as an exercise
in determining an insurance premium. Our more general change of measure is
still structure preserving, however, risk is priced also in the sense that one
slows down the speed of mean reversion. Such a reduction allow the random
fluctuations of the spot and the stochastic volatility last longer under the
pricing measure than under the objective probability, and thus spreads out the risk.

Although our analysis has a clear focus on the stylized facts of the risk
premium in power markets, the proposed class of pricing measures is clearly
also relevant in other commodity markets. As already mentioned, markets like
weather and freight share some similarities with power in that the underlying
"spot" is not storable. Also in more classical commodity markets like oil and
gas there are evidences of stochastic volatility and spot prices following a
mean-reversion dynamics, at least as a component of the spot. Moreover, in the
arithmetic case our analysis relates to the concept of unspanned volatility in
commodity markets, extensively studied by Trolle and Schwartz~\cite{TS}. The
forward price will not depend on the stochastic volatility factor, and hence
one cannot hedge options by forwards alone. Interestingly, the corresponding
geometric model will in fact span the stochastic volatility.

We present our results as follows. In the next Section we present the spot
model, and follow up in Section 3 with introducing our pricing measure
validating that this is indeed an equivalent probability. In Section 4 we
derive forward prices under the arithmetic spot price model, and analyse the
implied risk premium. Section 5 considers the corresponding forward prices and
the implied risk premium for the geometric spot price model. Here we exploit
the affine structure of the model to analyse the associated Riccati equation,
and provide insight into the potential risk premium profiles that our set-up
can generate. Both Section 4 and 5 have numerous empirical examples.

\section{Mathematical model}

Suppose that $(\Omega,\mathcal{F},\{\mathcal{F}_{t}\}_{t\in\lbrack0,T]},P) $
is a complete filtered probability space, where $T>0$ is a fixed finite time
horizon. On this probability space there are defined $W$, a standard Wiener
process, and $L,$ a pure jump L\'{e}vy subordinator with finite expectation,
that is a L\'{e}vy process with the following L\'{e}vy-It\^{o} representation
$L(t)=\int_{0}^{t}\int_{0}^{\infty}zN^{L}(ds,dz),t\in\lbrack0,T],$ where
$N^{L}(ds,dz)$ is a Poisson random measure with L\'{e}vy measure $\ell$
satisfying $\int_{0}^{\infty}z\ell(dz)<\infty.$ We shall suppose that $W$ and
$L$ are independent of each other.

As we are going to consider an Esscher change of measure and geometric spot
price models, we introduce the following assumption on the existence of
exponential moments of $L$.

\begin{assumption}
\label{Assumption_Exp_Theta_L}Suppose that%
\begin{equation}
\Theta_{L}\triangleq\sup\{\theta\in\mathbb{R}_{+}:\mathbb{E}[e^{\theta
L(1)}]<\infty\}, \label{Def_Theta_L}%
\end{equation}
is a constant strictly greater than one, which may be $\infty$.
\end{assumption}

Some remarks are in order.

\begin{remark}
\label{Remark_Cumulants}In $(-\infty,\Theta_{L})$ the cumulant (or log moment
generating) function $\kappa_{L}(\theta)\triangleq\log\mathbb{E}_{P}[e^{\theta
L(1)}]$ is well defined and analytic. As $0\in(-\infty,\Theta_{L})$, $L$ has
moments of all orders. Also, $\kappa_{L}(\theta)$ is convex, which yields that
$\kappa_{L}^{\prime\prime}(\theta)\geq0$ and, hence, that $\kappa_{L}^{\prime
}(\theta)$ is non decreasing. Finally, as a consequence of $L\geq0, a.s.$, we
have that $\kappa_{L}^{\prime}(\theta)$ is non negative.
\end{remark}

\begin{remark}
\label{Remark_D_Cumulants}Thanks to the L\'{e}vy-Kintchine representation of
$L$ we can express $\kappa_{L}(\theta)$ and its derivatives in terms of the
L\'{e}vy measure $\ell.$ We have that for $\theta\in(-\infty,\Theta_{L})$%
\begin{align*}
\kappa_{L}(\theta)  &  =\int_{0}^{\infty}(e^{\theta z}-1)\ell(dz)<\infty,\\
\kappa_{L}^{(n)}(\theta)  &  =\int_{0}^{\infty}z^{n}e^{\theta z}%
\ell(dz)<\infty,\quad n\in\mathbb{N},
\end{align*}
showing, in fact, that $\kappa_{L}^{(n)}(\theta)>0,n\in\mathbb{N}.$
\end{remark}

Consider the OU processes
\begin{align}
X(t)  &  =X(0)-\alpha\int_{0}^{t}X(s)ds+\int_{0}^{t}\sigma(s)dW(s)\quad
t\in\lbrack0,T],\label{Equ_OU_Brownian}\\
\sigma^{2}(t)  &  =\sigma^{2}(0)-\rho\int_{0}^{t}\sigma^{2}(s)ds+L(t),\quad
t\in\lbrack0,T], \label{Equ_OU_Levy}%
\end{align}
with $\alpha,\rho>0,X(0)\in\mathbb{R},\sigma^{2}(0)>0.$ Note that, in equation
\ref{Equ_OU_Brownian}, $X$ is written as a sum of a finite variation process
and a martingale. We may also rewrite equation \ref{Equ_OU_Levy} as a sum of a
finite variation part and pure jump martingale%
\[
\sigma^{2}(t)=\sigma^{2}(0)+\int_{0}^{t}(\kappa_{L}^{\prime}(0)-\rho\sigma
^{2}(s))ds+\int_{0}^{t}\int_{0}^{\infty}z\tilde{N}^{L}(ds,dz),\quad
t\in\lbrack0,T],
\]
where $\tilde{N}^{L}(ds,dz)\triangleq N^{L}(ds,dz)-ds\ \ell(dz)$ is the
compensated version of $N^{L}(ds,dz)$. In the notation of Shiryaev
\cite{Sh99}, page 669, the predictable characteristic triplets (with respect
to the pseudo truncation function $g(x)=x$) of $X$ and $\sigma^{2}$ are given
by
\[
(B^{X}(t),C^{X}(t),\nu^{X}(dt,dz))=(-\alpha\int_{0}^{t}X(s))ds,\int_{0}%
^{t}\sigma^{2}(s)ds,0),\quad t\in\lbrack0,T],
\]
and%
\[
(B^{\sigma^{2}}(t),C^{\sigma^{2}}(t),\nu^{\sigma^{2}}(dt,dz))=(\int_{0}%
^{t}(\kappa_{L}^{\prime}(0)-\rho\sigma^{2}(s))ds,0,\ell(dz)dt),\quad
t\in\lbrack0,T],
\]
respectively. In addition, applying It\^{o}'s Formula to $e^{\alpha t}X(t)$
and $e^{\rho t}\sigma^{2}(t),$ one can find the following explicit expressions
for $X(t)$ and $\sigma^{2}(t)$%
\begin{align}
X(t)  &  =X(s)e^{-\alpha(t-s)}+\int_{s}^{t}\sigma(u)e^{-\alpha(t-u)}%
dW(u),\label{Equ_X_Explicit_P}\\
\sigma^{2}(t)  &  =\sigma^{2}(s)e^{-\rho(t-s)}+\frac{\kappa_{L}^{\prime}%
(0)}{\rho}(1-e^{-\rho(t-s)})+\int_{s}^{t}\int_{0}^{\infty}e^{-\rho
(t-u)}z\tilde{N}^{L}(du,dz), \label{Equ_Sigma_Explicit_P}%
\end{align}
where $0\leq s\leq t\leq T.$

Using the notation in Kallsen and Muhle-Karbe \cite{KaMu-ka10}, we have that
the process $Z=(Z^{1}(t),Z_{2}(t))\triangleq(\sigma^{2}(t),X(t))$ has affine
differential characteristics given by%
\begin{align*}
\beta_{0}  &  =\left(
\begin{array}
[c]{c}%
\kappa_{L}^{\prime}(0)\\
0
\end{array}
\right)  ,\quad\gamma_{0}=\left(
\begin{array}
[c]{cc}%
0 & 0\\
0 & 0
\end{array}
\right)  ,\quad\varphi_{0}(A)=\int_{0}^{\infty}\boldsymbol{1}_{A}%
(z,0)\ell(dz),\forall A\in\mathcal{B}(\mathbb{R}^{2})\\
\beta_{1}  &  =\left(
\begin{array}
[c]{c}%
-\rho\\
0
\end{array}
\right)  ,\quad\gamma_{1}=\left(
\begin{array}
[c]{cc}%
0 & 0\\
0 & 1
\end{array}
\right)  ,\quad\varphi_{1}(A)\equiv0,\forall A\in\mathcal{B}(\mathbb{R}%
^{2}),\\
\beta_{2}  &  =\left(
\begin{array}
[c]{c}%
0\\
-\alpha
\end{array}
\right)  ,\quad\gamma_{2}=\left(
\begin{array}
[c]{cc}%
0 & 0\\
0 & 0
\end{array}
\right)  ,\quad\varphi_{2}(A)\equiv0,\forall A\in\mathcal{B}(\mathbb{R}^{2}).
\end{align*}
These characteristics are admissible and correspond to an affine process in
$\mathbb{R}_{+}\times\mathbb{R}.$

\section{The change of measure\label{SecChangeOfMeasure}}

We will consider a parametrized family of measure changes which will allow us
to simultaneously modify the speed and the level of mean reversion in
equations \eqref{Equ_OU_Brownian} and \eqref{Equ_OU_Levy}. The density
processes of these measure changes will be determined by the stochastic
exponential of certain martingales. To this end, consider the following family
of kernels%
\begin{align}
G_{\theta_{1},\beta_{1}}(t)  &  \triangleq\sigma^{-1}(t)\left(  \theta
_{1}+\alpha\beta_{1}X(t)\right)  ,\quad t\in\lbrack0,T],\label{Equ_G}\\
H_{\theta_{2},\beta_{2}}(t,z)  &  \triangleq e^{\theta_{2}z}\left(
1+\frac{\rho\beta_{2}}{\kappa_{L}^{\prime\prime}(\theta_{2})}z\sigma
^{2}(t-)\right)  ,\quad t\in\lbrack0,T],z\in\mathbb{R}. \label{Equ_H}%
\end{align}
The parameters $\bar{\beta}\triangleq(\beta_{1},\beta_{2})$ and $\bar{\theta
}\triangleq(\theta_{1},\theta_{2})$ will take values on the following sets
$\bar{\beta}\in\lbrack0,1]^{2},\bar{\theta}\in\bar{D}_{L}\triangleq
\mathbb{R}\times D_{L},$where $D_{L}\triangleq(-\infty,\Theta_{L}/2)$ and
$\Theta_{L}$ is given by equation \eqref{Def_Theta_L}. By
Assumption~\ref{Assumption_Exp_Theta_L} and Remarks \ref{Remark_Cumulants} and
\ref{Remark_D_Cumulants}, these kernels are well defined.

\begin{example}
\label{Example_Subordinators} Typical examples of $\ell,\Theta_{L}$ and
$D_{L}$ are the following:

\begin{enumerate}
\item Bounded support: $L$ has a jump of size $a,$ i.e. $\ell=\delta_{a}.$ In
this case $\Theta_{L}=\infty$ and $D_{L}=\mathbb{R}.$

\item Finite activity: $L$ is a compound Poisson process with exponential
jumps, i.e., $\ell(dz)=ce^{-\lambda z}1_{(0,\infty)}\allowbreak dz,$ for some
$c>0$ and $\lambda>0.$ In this case $\Theta_{L}=\lambda$ and $D_{L}%
=(-\infty,\lambda/2).$

\item Infinite activity: $L$ is a tempered stable subordinator, i.e.,
$\ell(dz)=cz^{-(1+\alpha)}\allowbreak e^{-\lambda z}\allowbreak1_{(0,\infty
)}dz,$ for some $c>0,\lambda>0$ and $\alpha\in\lbrack0,1).$ In this case also
$\Theta_{L}=\lambda$ and $D_{L}=(-\infty,\lambda/2).$
\end{enumerate}
\end{example}

Next, for $\bar{\beta}\in\lbrack0,1]^{2},\bar{\theta}\in\bar{D}_{L},$ define
the following family of Wiener and Poisson integrals%
\begin{align*}
\tilde{G}_{\theta_{1},\beta_{1}}(t)  &  \triangleq\int_{0}^{t}G_{\theta
_{1},\beta_{1}}(s)dW(s),\quad t\in\lbrack0,T],\\
\tilde{H}_{\theta_{2},\beta_{2}}(t)  &  \triangleq\int_{0}^{t}\int_{0}%
^{\infty}\left(  H_{\theta_{2},\beta_{2}}(s,z)-1\right)  \tilde{N}%
^{L}(ds,dz),\quad t\in\lbrack0,T],
\end{align*}
associated to the kernels $G_{\theta_{1},\beta_{1}}$ and $H_{\theta_{2}%
,\beta_{2}},$ respectively.

We propose a family of measure changes given by $Q_{\bar{\theta},\bar{\beta}%
}\sim P,\bar{\beta}\in\lbrack0,1]^{2},\bar{\theta}\in\bar{D}_{L},$ with%
\begin{equation}
\left.  \frac{dQ_{\bar{\theta},\bar{\beta}}}{dP}\right\vert _{\mathcal{F}_{t}%
}\triangleq\mathcal{E}(\tilde{G}_{\theta_{1},\beta_{1}}+\tilde{H}_{\theta
_{2},\beta_{2}})(t),\quad t\in\lbrack0,T]. \label{EquDefQ}%
\end{equation}
Let us assume for a moment that $\mathcal{E}(\tilde{G}_{\theta_{1},\beta_{1}%
}+\tilde{H}_{\theta_{2},\beta_{2}})$ is a strictly positive true martingale
(this will be proven in Theorem~\ref{thm:true_martingale} below): Then, by
Girsanov's theorem for semimartingales (Theorems 1 and 3, pages 702 and 703 in
Shiryaev \cite{Sh99}), the process $X(t)$ and $\sigma^{2}(t)$ become
\begin{align*}
X(t)  &  =X(0)+B_{Q_{\bar{\theta},\bar{\beta}}}^{X}(t)+\sigma(t)W_{Q_{\bar
{\theta},\bar{\beta}}}(t),\quad t\in\lbrack0,T],\\
\sigma^{2}(t)  &  =\sigma^{2}(0)+B_{Q_{\bar{\theta},\bar{\beta}}}^{\sigma^{2}%
}(t)+\int_{0}^{t}\int_{0}^{\infty}z\tilde{N}_{Q_{\bar{\theta},\bar{\beta}}%
}^{L}(ds,dz),\quad t\in\lbrack0,T],
\end{align*}
with%
\begin{align}
B_{Q_{\bar{\theta},\bar{\beta}}}^{X}(t)  &  =\int_{0}^{t}(\theta_{1}%
-\alpha(1-\beta_{1})X(s))ds,\quad t\in\lbrack0,T],\label{Equ_X_Drift_Q}\\
B_{Q_{\bar{\theta},\bar{\beta}}}^{\sigma^{2}}(t)  &  =\int_{0}^{t}(\kappa
_{L}^{\prime}(0)-\rho\sigma^{2}(s))ds+\int_{0}^{t}\int_{0}^{\infty}%
z(H_{\theta_{2},\beta_{2}}(s,z)-1)\ell(dz)ds\label{Equ_Sig2_Drift_Q}\\
&  =\int_{0}^{t}\{(\kappa_{L}^{\prime}(0)-\rho\sigma^{2}(s))+\int_{0}^{\infty
}z(e^{\theta_{2}z}-1)\ell(dz)\nonumber\\
&  \qquad+\frac{\rho\beta_{2}}{\kappa_{L}^{\prime\prime}(\theta_{2})}\int
_{0}^{\infty}z^{2}e^{\theta_{2}z}\ell(dz)\sigma^{2}(s-)\}ds\nonumber\\
&  =\int_{0}^{t}\left(  \kappa_{L}^{\prime}(\theta_{2})-\rho(1-\beta
_{2})\sigma^{2}(s)\right)  ds,\quad t\in\lbrack0,T],\nonumber
\end{align}
where $W_{Q_{\bar{\theta},\bar{\beta}}}$ is a $Q_{\bar{\theta},\bar{\beta}}%
$-standard Wiener process and the $Q_{\bar{\theta},\bar{\beta}}$-compensator
measure of $\sigma^{2}$ (and $L$) is
\[
v_{Q_{\bar{\theta},\bar{\beta}}}^{\sigma^{2}}(dt,dz)=v_{Q_{\bar{\theta}%
,\bar{\beta}}}^{L}(dt,dz)=H_{\theta_{2},\beta_{2}}(t,z)\ell(dz)dt.
\]
In conclusion, the semimartingale triplet for $X$ and $\sigma^{2}$ under
$Q_{\bar{\theta},\bar{\beta}}$ are given by $(B_{Q_{\bar{\theta},\bar{\beta}}%
}^{X},\int_{0}^{\cdot}\sigma^{2}(s)ds,0)$ and $(B_{Q_{\bar{\theta},\bar{\beta
}}}^{\sigma^{2}},0,v_{Q_{\bar{\theta},\bar{\beta}}}^{\sigma^{2}}),$ respectively.

\begin{remark}
\label{RemarkAffineCharactUnderQ}Under $Q_{\bar{\theta},\bar{\beta}},$ the
process $Z=(\sigma^{2}(t),X(t))$ is affine with differential characteristics
given by%
\begin{align*}
\beta_{0}  &  =\left(
\begin{array}
[c]{c}%
\kappa_{L}^{\prime}(\theta_{2})\\
\theta_{1}%
\end{array}
\right)  ,\quad\gamma_{0}=\left(
\begin{array}
[c]{cc}%
0 & 0\\
0 & 0
\end{array}
\right)  ,\quad\varphi_{0}(A)=\int_{0}^{\infty}\boldsymbol{1}_{A}%
(z,0)e^{\theta_{2}z}\ell(dz),\forall A\in\mathcal{B}(\mathbb{R}^{2}),\\
\beta_{1}  &  =\left(
\begin{array}
[c]{c}%
-\rho(1-\beta_{2})\\
0
\end{array}
\right)  ,\quad\gamma_{1}=\left(
\begin{array}
[c]{cc}%
0 & 0\\
0 & 1
\end{array}
\right)  ,\quad\varphi_{1}(A)=\int_{0}^{\infty}\boldsymbol{1}_{A}%
(z,0)\frac{\rho\beta_{2}}{\kappa_{L}^{\prime\prime}(\theta_{2})}ze^{\theta
_{2}z}\ell(dz),\forall A\in\mathcal{B}(\mathbb{R}^{2}),\\
\beta_{2}  &  =\left(
\begin{array}
[c]{c}%
0\\
-\alpha(1-\beta_{1})
\end{array}
\right)  ,\quad\gamma_{2}=\left(
\begin{array}
[c]{cc}%
0 & 0\\
0 & 0
\end{array}
\right)  ,\quad\varphi_{2}(A)\equiv0,\forall A\in\mathcal{B}(\mathbb{R}^{2}).
\end{align*}
These characteristics are admissible and correspond to an affine process in
$\mathbb{R}_{+}\times\mathbb{R}.$
\end{remark}

\begin{remark}
\label{Remark_Dynamics_Q}Under $Q_{\bar{\theta},\bar{\beta}},$ $\sigma^{2}$
still satisfies the Langevin equation with different parameters, that is, the
measure change preserves the structure of the equations for $\sigma^{2}$.
However, the process $L$ is not a L\'{e}vy process under $Q_{\bar{\theta}%
,\bar{\beta}}$, but it remains a semimartingale. The equation for $X$ is the
same under the new measure but with different parameters. Therefore, one can
use It\^{o}'s Formula again to obtain the following explicit expressions for
$X$ and $\sigma^{2}$%
\begin{align}
X(t)  &  =X(s)e^{-\alpha(1-\beta_{1})(t-s)}+\frac{\theta_{1}}{\alpha
(1-\beta_{1})}(1-e^{-\alpha(1-\beta_{1})(t-s)})\label{EquXDynam-Q}\\
&  \qquad+\int_{s}^{t}\sigma(u)e^{-\alpha(1-\beta_{1})(t-u)}dW_{Q_{\bar
{\theta},\bar{\beta}}}(u),\nonumber\\
\sigma^{2}(t)  &  =\sigma^{2}(s)e^{-\rho(1-\beta_{2})(t-s)}+\frac{\kappa
_{L}^{\prime}(\theta_{2})}{\rho(1-\beta_{2})}(1-e^{-\rho(1-\beta_{2}%
)(t-s)})\label{EquSigmaDynam-Q}\\
&  \qquad+\int_{s}^{t}\int_{0}^{\infty}e^{-\rho(1-\beta_{2})(t-u)}z\tilde
{N}_{Q_{\theta,\beta}}^{L}(du,dz),\nonumber
\end{align}
where $0\leq s\leq t\leq T.$
\end{remark}

We prove that $Q_{\bar{\theta},\bar{\beta}}$ is a true probability measure,
that is, $\mathcal{E}(\tilde{G}_{\theta_{1},\beta_{1}}+\tilde{H}_{\theta
_{2},\beta_{2}})(t)$ is a strictly positive true martingale under $P$ for
$t\leq T$. We have the following theorem.

\begin{theorem}
\label{thm:true_martingale} Let $\theta\in\bar{D}_{L},\bar{\beta}\in
\lbrack0,1]^{2}$. Then $\mathcal{E}(\tilde{G}_{\theta_{1},\beta_{1}}+\tilde
{H}_{\theta_{2},\beta_{2}})=\{\mathcal{E}(\tilde{G}_{\theta_{1},\beta_{1}%
}+\tilde{H}_{\theta_{2},\beta_{2}})(t)\}_{t\in\lbrack0,T]}$ is a strictly
positive true martingale under $P$.
\end{theorem}

\begin{proof}
That $\mathcal{E}(\tilde{G}_{\theta_{1},\beta_{1}}+\tilde{H}_{\theta_{2}%
,\beta_{2}})$ is strictly positive follows easily from the fact that the
L\'{e}vy process $L$ is a subordinator as this yields strictly positive jumps
of $\tilde{G}_{\theta_{1},\beta_{1}}+\tilde{H}_{\theta_{2},\beta_{2}}$. It
holds that $[\tilde{G}_{\theta_{1},\beta_{1}},\tilde{H}_{\theta_{2},\beta_{2}%
}],$ the quadratic co-variation between $\tilde{G}_{\theta_{1},\beta_{1}}$ and
$\tilde{H}_{\theta_{2},\beta_{2}},$ is identically zero, by Yor's formula in
Protter~\cite[Theorem 38]{Protter04}. Hence, we can write%
\begin{equation}
\mathcal{E}(\tilde{G}_{\theta_{1},\beta_{1}}+\tilde{H}_{\theta_{2},\beta_{2}%
})(t)=\mathcal{E}(\tilde{G}_{\theta_{1},\beta_{1}})(t)\mathcal{E}(\tilde
{H}_{\theta_{2},\beta_{2}})(t),\quad t\in\lbrack0,T].\label{Equ_Yor_Formula}%
\end{equation}
By classical martingale theory, we know that $\mathcal{E}(\tilde{G}%
_{\theta_{1},\beta_{1}}+\tilde{H}_{\theta_{2},\beta_{2}})$ is a true
martingale if and only if
\[
\mathbb{E}_{P}[\mathcal{E}(\tilde{G}_{\theta_{1},\beta_{1}}+\tilde{H}%
_{\theta_{2},\beta_{2}})(T)]=1,
\]
which, using Yor's formula, is equivalent to showing that%
\[
\mathbb{E}_{P}[\mathcal{E}(\tilde{G}_{\theta_{1},\beta_{1}})(T)\mathcal{E}%
(\tilde{H}_{\theta_{2},\beta_{2}})(T)]=1.
\]
Let, $\mathcal{F}_{T}^{L}$ be the $\sigma$-algebra generated by $L$ up to time
$T,$ then we have that%
\[
\mathbb{E}_{P}[\mathcal{E}(\tilde{G}_{\theta_{1},\beta_{1}})(T)\mathcal{E}%
(\tilde{H}_{\theta_{2},\beta_{2}})(T)]=\mathbb{E}_{P}[\mathbb{E}%
_{P}[\mathcal{E}(\tilde{G}_{\theta_{1},\beta_{1}})(T)|\mathcal{F}_{T}%
^{L}]\mathcal{E}(\tilde{H}_{\theta_{2},\beta_{2}})(T)].
\]
If we show that $\mathbb{E}_{P}[\mathcal{E}(\tilde{G}_{\theta_{1},\beta_{1}%
})(T)|\mathcal{F}_{T}^{L}]\equiv1,$ then we will have finished, because by
Theorem 3.10 in Benth and Ortiz-Latorre \cite{BeO-L13}, we have that
$\mathcal{E}(\tilde{H}_{\theta_{2},\beta_{2}})$ is a true martingale and,
hence, $\mathbb{E}_{P}[\mathcal{E}(\tilde{H}_{\theta_{2},\beta_{2}})(T)]=1.$
The idea of the proof is based on the fact that $\mathbb{E}_{P}[\mathcal{E}%
(\tilde{G}_{\theta_{1},\beta_{1}})(T)|\mathcal{F}_{T}^{L}]$ is the expectation
of $\mathcal{E}(\tilde{G}_{\theta_{1},\beta_{1}})(T)$ assuming that
$\sigma(t)$ is a deterministic function that, in addition, is bounded below by
$\sigma(0)e^{-\rho t}.$ Using this information one can show that,
conditionally on knowing $\sigma,\mathcal{E}(\tilde{G}_{\theta_{1},\beta_{1}%
})$ is a true martingale and, hence, $\mathbb{E}_{P}[\mathcal{E}(\tilde
{G}_{\theta_{1},\beta_{1}})(T)]=1.$ Let us sketch the proof that is basically
the same as in Section 3.1 in \cite{BeO-L13} but, now, with $\sigma$ being a
function. First, we show that, conditionally on $\mathcal{F}_{T}^{L},$
$\tilde{G}_{\theta_{1},\beta_{1}}$ is a square integrable $P$-martingale
because
\begin{align*}
\mathbb{E}_{P}[(\tilde{G}_{\theta_{1},\beta_{1}})^{2}|\mathcal{F}_{T}^{L}] &
=\mathbb{E}_{P}[\int_{0}^{T}\sigma^{-2}(t)\left(  \theta_{1}+\alpha\beta
_{1}X(t)\right)  ^{2}dt|\mathcal{F}_{T}^{L}]\\
&  \leq2\sigma(0)^{-2}e^{2\rho T}\left(  \theta_{1}^{2}T+\alpha^{2}%
\mathbb{E}_{P}[\int_{0}^{T}X^{2}(t)dt]\right)  <\infty,
\end{align*}
(see Proposition 3.6. in \cite{BeO-L13}). To show that $\mathcal{E}(\tilde
{G}_{\theta_{1},\beta_{1}})$ is a $P$-martingale on $[0,T],$ we consider a
reducing sequence of stopping times $\{\tau_{n}\}_{n\geq1}$ for $\mathcal{E}%
(\tilde{G}_{\theta_{1},\beta_{1}})$ and, proceeding as in Theorem 3.7 in
\cite{BeO-L13}, we define a sequence of probability measure $\{Q_{\theta
_{1},\beta_{1}}^{n}\}_{n\geq1}$ with Radon-Nykodim densities given by $\left.
\frac{dQ_{\theta_{1},\beta_{1}}^{n}}{dP}\right\vert _{\mathcal{F}_{t}%
}\triangleq\mathcal{E}(\tilde{G}_{\theta_{1},\beta_{1}})^{\tau_{n}}%
(t),t\in\lbrack0,T],n\geq1.$ Doing the same reasonings as in Theorem 3.7 in
\cite{BeO-L13}, we reduce the problem to prove that
\[
\sup_{n\geq1}\mathbb{E}_{Q_{\theta_{1},\beta_{1}}^{n}}[\int_{0}^{T}%
\boldsymbol{1}_{[0,\tau_{n}]}(G_{\theta_{1},\beta_{1}}(t))^{2}dt]<\infty.
\]
Now, one has
\[
\mathbb{E}_{Q_{\theta_{1},\beta_{1}}^{n}}[\int_{0}^{T}\boldsymbol{1}%
_{[0,\tau_{n}]}(G_{\theta_{1},\beta_{1}}(t))^{2}dt]\leq2\sigma(0)^{-2}e^{2\rho
T}\left(  \theta_{1}^{2}T+\alpha^{2}\mathbb{E}_{Q_{\theta_{1},\beta_{1}}^{n}%
}[\int_{0}^{T}\boldsymbol{1}_{[0,\tau_{n}]}X^{2}(t)dt]\right)  .
\]
To bound the last expectation in the previous expression we use that we know
the dynamics of $X(t)$ for $t\in\lbrack0,\tau_{n}]$ under $Q_{\theta_{1}%
,\beta_{1}}^{n}$, which is obtained from equation \eqref{EquXDynam-Q} by
setting $s=0$ and $t<\tau_{n}.$ Therefore, we can write%
\begin{align*}
&  \mathbb{E}_{Q_{\theta_{1},\beta_{1}}^{n}}[\int_{0}^{T}\boldsymbol{1}%
_{[0,\tau_{n}]}(t)X(t)^{2}dt]\\
&  \qquad\leq2\left\{  \mathbb{E}_{Q_{\theta_{1},\beta_{1}}^{n}}[\int_{0}%
^{T}\boldsymbol{1}_{[0,\tau_{n}]}(t)\left(  X(0)e^{-\alpha(1-\beta_{1}%
)t}+\frac{\theta_{1}}{\alpha(1-\beta_{1})}\left(  1-e^{-\alpha(1-\beta_{1}%
)t}\right)  \right)  ^{2}dt]\right.  \\
&  \qquad\qquad\left.  +\mathbb{E}_{Q_{\theta_{1},\beta_{1}}^{n}}[\int_{0}%
^{T}\boldsymbol{1}_{[0,\tau_{n}]}(t)\left(  \int_{0}^{t}\sigma(s)e^{-\alpha
(1-\beta_{1})(t-s)}dW_{Q_{\theta_{1},\beta_{1}}^{n}}(s)\right)  ^{2}%
dt]\right\}  \\
&  \qquad\leq2T\{\left(  \left\vert X(0)\right\vert +\left(  \left\vert
\theta_{1}\right\vert \right)  T\right)  ^{2}+\sigma(0)^{-2}e^{2\rho T}%
T^{2}\}<\infty.
\end{align*}
Here, we have used that the function $\eta(x)\triangleq(1-e^{-xa})/x\leq a$
for $x,a\geq0,$ and that
\begin{align*}
&  \mathbb{E}_{Q_{\theta_{1},\beta_{1}}^{n}}[\left(  \int_{0}^{t}%
\sigma(s)e^{-\alpha(1-\beta_{1})(t-s)}dW_{Q_{\theta_{1},\beta_{1}}^{n}%
}(s)\right)  ^{2}]\\
&  \qquad\qquad=\sigma(0)^{-2}e^{2\rho T}\int_{0}^{t}e^{-2\alpha(1-\beta
_{1})(t-s)}ds\leq\sigma(0)^{-2}e^{2\rho T}T.
\end{align*}
The Theorem follows.
\end{proof}

We also have the following result on the independence of the driving noise
processes after the change of measure:

\begin{lemma}
Under $Q_{\bar{\theta},\bar{\beta}}$, the Brownian motion $W_{Q_{\bar{\theta
},\bar{\beta}}}$ and the random measure $N_{Q_{\bar{\theta},\bar{\beta}}}^{L}$
are independent.
\end{lemma}

\begin{proof}
To prove the independence of $W_{Q_{\bar{\theta},\bar{\beta}}}$ and
$N_{Q_{\bar{\theta},\bar{\beta}}}^{L}$ under $Q_{\bar{\theta},\bar{\beta}},$
it is sufficient to prove that%
\begin{align*}
&  \mathbb{E}_{Q_{\bar{\theta},\bar{\beta}}}[\exp(\mathrm{i}\sum_{j=1}%
^{k}\left(  \mu_{j}W_{Q_{\bar{\theta},\bar{\beta}}}(t_{j})+\xi_{j}\int
_{0}^{t_{j}}\int_{0}^{\infty}z\tilde{N}_{Q_{\bar{\theta},\bar{\beta}}}%
^{L}(ds,dz)\right)  )]\\
&  \qquad=\mathbb{E}_{Q_{\bar{\theta},\bar{\beta}}}[\exp(\mathrm{i}\sum
_{j=1}^{k}\mu_{j}W_{Q_{\bar{\theta},\bar{\beta}}}(t_{j}))]\mathbb{E}%
_{Q_{\bar{\theta},\bar{\beta}}}[\exp(\mathrm{i}\sum_{j=1}^{k}\xi_{j}\int
_{0}^{t_{j}}\int_{0}^{\infty}z\tilde{N}_{Q_{\bar{\theta},\bar{\beta}}}%
^{L}(ds,dz))],
\end{align*}
for any $\mu_{j},\xi_{j}\in\mathbb{R},j=1,...,k$ and $0\leq t_{1}<t_{2}%
<\cdots<t_{k-1}<t_{k}\leq T.$ We will make use of the following notation:
given a process $Z=\{Z(t)\}_{t\in\lbrack0,T]}$ we will denote by $\Delta
_{j}Z=Z(t_{j})-Z(t_{j-1}),j=1,...,k,$ where $t_{0}=0,$ by convention. We have
that%
\begin{align*}
&  \mathbb{E}_{Q_{\bar{\theta},\bar{\beta}}}[\exp(\mathrm{i}\sum_{j=1}%
^{k}\left(  \mu_{j}W_{Q_{\bar{\theta},\bar{\beta}}}(t_{j})+\xi_{j}\int
_{0}^{t_{j}}\int_{0}^{\infty}z\tilde{N}_{Q_{\bar{\theta},\bar{\beta}}}%
^{L}(ds,dz)\right)  )]\\
&  =\mathbb{E}_{P}[\mathcal{E}(\tilde{G}_{\theta_{1},\beta_{1}}+\tilde
{H}_{\theta_{2},\beta_{2}})(t_{k})\exp(\mathrm{i}\sum_{j=1}^{k}\left(  \mu
_{j}W_{Q_{\bar{\theta},\bar{\beta}}}(t_{j})+\xi_{j}\int_{0}^{t_{j}}\int
_{0}^{\infty}z\tilde{N}_{Q_{\bar{\theta},\bar{\beta}}}^{L}(ds,dz)\right)  )]\\
&  =\mathbb{E}_{P}[\mathbb{E}_{P}[\mathcal{E}(\tilde{G}_{\theta_{1},\beta_{1}%
})(t_{k})\exp(\mathrm{i}\sum_{j=1}^{k}\mu_{j}W_{Q_{\bar{\theta},\bar{\beta}}%
}(t_{j}))|\mathcal{F}_{t_{k}}^{L}]\mathcal{E}(\tilde{H}_{\theta_{2},\beta_{2}%
})(t_{k})\exp(\mathrm{i}\sum_{j=1}^{k}\xi_{j}\int_{0}^{t_{j}}\int_{0}^{\infty
}z\tilde{N}_{Q_{\bar{\theta},\bar{\beta}}}^{L}(ds,dz))],
\end{align*}
and%
\begin{align*}
&  \mathbb{E}_{P}[\mathcal{E}(\tilde{G}_{\theta_{1},\beta_{1}})(t_{k}%
)\exp(\mathrm{i}\sum_{j=1}^{k}\mu_{j}W_{Q_{\bar{\theta},\bar{\beta}}}%
(t_{j}))|\mathcal{F}_{t_{k}}^{L}]\\
&  =\mathbb{E}_{P}[\mathcal{E}(\tilde{G}_{\theta_{1},\beta_{1}})(t_{k-1}%
)\exp(\int_{t_{k-1}}^{t_{k}}G_{\theta_{1},\beta_{1}}(s)dW(s)-\frac{1}{2}%
\int_{t_{k-1}}^{t_{k}}G_{\theta_{1},\beta_{1}}^{2}(s)ds)\\
&  \qquad\times\exp(\mathrm{i}\sum_{j=1}^{k-1}\mu_{j}W_{Q_{\bar{\theta}%
,\bar{\beta}}}(t_{j})+\mathrm{i}\mu_{k}W_{Q_{\bar{\theta},\bar{\beta}}%
}(t_{k-1})+\mathrm{i}\mu_{k}\Delta W_{Q_{\bar{\theta},\bar{\beta}}%
})|\mathcal{F}_{t_{k}}^{L}]\\
&  =\mathbb{E}_{P}[\mathcal{E}(\tilde{G}_{\theta_{1},\beta_{1}})(t_{k-1}%
)\exp(\mathrm{i}\sum_{j=1}^{k-1}\mu_{j}W_{Q_{\bar{\theta},\bar{\beta}}}%
(t_{j})+\mathrm{i}\mu_{k}W_{Q_{\bar{\theta},\bar{\beta}}}(t_{k-1}))\\
&  \qquad\times\mathbb{E}_{P}[\exp(\int_{t_{k-1}}^{t_{k}}G_{\theta_{1}%
,\beta_{1}}(s)dW(s)-\frac{1}{2}\int_{t_{k-1}}^{t_{k}}G_{\theta_{1},\beta_{1}%
}^{2}(s)ds\\
&  \qquad+\mathrm{i}\mu_{k}\left(  \Delta_{k}W-\int_{t_{k-1}}^{t_{k}}%
G_{\theta_{1},\beta_{1}}(s)ds\right)  )|\mathcal{F}_{t_{k}}^{L}\vee
\mathcal{F}_{t_{k-1}}^{W}]|\mathcal{F}_{t_{k}}^{L}].
\end{align*}
Moreover, using similar arguments to those used in the proof of Theorem
\ref{thm:true_martingale}, we have that
\[
\exp(\int_{0}^{t}(G_{\theta_{1},\beta_{1}}(s)+\mathrm{i}\mu_{k})dW(s)-\frac
{1}{2}\int_{0}^{t}(G_{\theta_{1},\beta_{1}}(s)+\mathrm{i}\mu_{k})^{2}ds),
\]
is a $\mathcal{F}_{t_{k}}^{L}\vee\mathcal{F}_{t}^{W}$-martingale and, then, we
get
\begin{align*}
&  \mathbb{E}_{P}[\exp(\int_{t_{k-1}}^{t_{k}}G_{\theta_{1},\beta_{1}%
}(s)dW(s)-\frac{1}{2}\int_{t_{k-1}}^{t_{k}}G_{\theta_{1},\beta_{1}}%
^{2}(s)ds+\mathrm{i}\mu_{k}\left(  \Delta_{k}W-\int_{t_{k-1}}^{t_{k}}%
G_{\theta_{1},\beta_{1}}(s)ds\right)  )|\mathcal{F}_{t_{k}}^{L}\vee
\mathcal{F}_{t_{k-1}}^{W}]\\
&  =e^{-\frac{1}{2}\mu_{k}^{2}\Delta_{k}t}\mathbb{E}_{P}[\exp(\int_{t_{k-1}%
}^{t_{k}}(G_{\theta_{1},\beta_{1}}(s)+\mathrm{i}\mu_{k})dW(s)-\frac{1}{2}%
\int_{t_{k-1}}^{t_{k}}(G_{\theta_{1},\beta_{1}}(s)+\mathrm{i}\mu_{k}%
)^{2}ds)|\mathcal{F}_{t_{k}}^{L}\vee\mathcal{F}_{t_{k-1}}^{W}]\\
&  =e^{-\frac{1}{2}\mu_{k}^{2}\Delta_{k}t}.
\end{align*}
Therefore,%
\begin{align*}
&  \mathbb{E}_{P}[\mathcal{E}(\tilde{G}_{\theta_{1},\beta_{1}})(t_{k}%
)\exp(\mathrm{i}\sum_{j=1}^{k}\mu_{j}W_{Q_{\bar{\theta},\bar{\beta}}}%
(t_{j}))|\mathcal{F}_{t_{k}}^{L}]\\
&  =e^{-\frac{1}{2}\mu_{k}^{2}\Delta_{k}t}\mathbb{E}_{P}[\mathcal{E}(\tilde
{G}_{\theta_{1},\beta_{1}})(t_{k-1})\exp(\mathrm{i}\sum_{j=1}^{k-1}\mu
_{j}W_{Q_{\bar{\theta},\bar{\beta}}}(t_{j})+\mathrm{i}\mu_{k}W_{Q_{\bar
{\theta},\bar{\beta}}}(t_{k-1}))|\mathcal{F}_{t_{k}}^{L}].
\end{align*}
Repeating the previous conditioning trick, one gets that%
\[
\mathbb{E}_{P}[\mathcal{E}(\tilde{G}_{\theta_{1},\beta_{1}})(t_{k}%
)\exp(\mathrm{i}\sum_{j=1}^{k}\mu_{j}W_{Q_{\bar{\theta},\bar{\beta}}}%
(t_{j}))|\mathcal{F}_{t_{k}}^{L}]=\exp\left(  -\frac{1}{2}\sum_{j=1}%
^{k}\left(  \sum_{q=j}^{k}\mu_{q}^{2}\right)  \Delta_{j}t\right)  .
\]
and, therefore,
\begin{align*}
&  \mathbb{E}_{Q_{\bar{\theta},\bar{\beta}}}[\exp(\mathrm{i}\sum_{j=1}%
^{k}\left(  \mu_{j}W_{Q_{\bar{\theta},\bar{\beta}}}(t_{j})+\xi_{j}\int
_{0}^{t_{j}}\int_{0}^{\infty}z\tilde{N}_{Q_{\bar{\theta},\bar{\beta}}}%
^{L}(ds,dz)\right)  )]\\
&  =\exp\left(  -\frac{1}{2}\sum_{j=1}^{k}\left(  \sum_{q=j}^{k}\mu
_{q}\right)  ^{2}\Delta_{j}t\right)  \mathbb{E}_{Q_{\bar{\theta},\bar{\beta}}%
}[\exp(\mathrm{i}\sum_{j=1}^{k}\xi_{j}\int_{0}^{t_{j}}\int_{0}^{\infty}%
z\tilde{N}_{Q_{\bar{\theta},\bar{\beta}}}^{L}(ds,dz))]
\end{align*}
On the other hand,%
\begin{align*}
\mathbb{E}_{Q_{\bar{\theta},\bar{\beta}}}[\exp(\mathrm{i}\sum_{j=1}^{k}\mu
_{j}W_{Q_{\bar{\theta},\bar{\beta}}}(t_{j}))] &  =\mathbb{E}_{Q_{\bar{\theta
},\bar{\beta}}}[\exp(\mathrm{i}\sum_{j=1}^{k}\left(  \sum_{q=j}^{k}\mu
_{q}\right)  \Delta_{j}W_{Q_{\bar{\theta},\bar{\beta}}})]\\
&  =\exp\left(  -\frac{1}{2}\sum_{j=1}^{k}\left(  \sum_{q=j}^{k}\mu
_{q}\right)  ^{2}\Delta_{j}t\right)  ,
\end{align*}
and we can conclude the proof.
\end{proof}

One of the particularities of electricity markets is that power is a non
storable commodity and for that reason is not a directly tradeable financial
asset. This entails that one can not derive the forward price of electricity
from the classical buy-and-hold hedging argument. Using a risk-neutral pricing
argument (see Benth, \v{S}altyt\.{e} Benth and Koekebakker \cite{BeSa-BeKo08}%
), under the assumption of deterministic interest rates, the forward price at
time $0\leq t$, with time of delivery $T$ with $t\leq T<T^{\ast},$ is given by
$F_{Q}(t,T)\triangleq\mathbb{E}_{Q}[S(T)|\mathcal{F}_{t}].$ Here, $Q$ is any
probability measure equivalent to the historical measure $P$ and
$\mathcal{F}_{t}$ is the market information up to time $t$. In what follows we
will use the probability measure $Q=Q_{\bar{\theta},\bar{\beta}}$ introduced
above, and let the spot price dynamics be given in terms of the process $X(t)$
and $\sigma^{2}(t)$ in \eqref{Equ_OU_Brownian}-\eqref{Equ_OU_Levy}. This will
provide us with a parametric class of structure-preserving probability
measures, extending the Esscher transform but still being reasonably
analytically tractable from a pricing point of view.

Our choice of pricing measure can also be applied to temperature futures
markets, where the underlying "asset" is a temperature index measured in some
location. Temperature is clearly not financially tradeable. There is empirical
evidence for mean-reversion and stochastic volatility in temperature data, see
Benth and \v{S}altyt\.{e} Benth~\cite{BSB-sv}. Yet another example is the
freight rate market, where the "spot" typically is an index obtained from
opinions of traders. See Benth, Koekebakker and Taib~\cite{BKT} for stochastic
modelling of freight rate spot data, with models of the form \eqref{Equ_OU_Brownian}-\eqref{Equ_OU_Levy}.

Oil and gas can typically be stored, and one can build a pricing model for
forwards by including storage and transportation costs, as well as the
convenience yield (see e.g. Eydeland and Wolyniec~\cite{EW} and
Geman~\cite{Geman}). However, we may also in this case use the probability
measure $Q=Q_{\bar{\theta},\bar{\beta}}$ as a parametric class of pricing
measures. Firstly, the underlying spot assets do not need to be (local)
martingales under the pricing measure in these markets, although the assets
are tradable, since there are frictions yielding market incompleteness.
Secondly, the probability measures provide a flexible way to model the risk
premium (as we shall see later), and therefore may be attractive over models
that directly specifies the dynamics of a convenience yield, say (see e.g.
Eydeland and Wolyniec~\cite{EW} for such models). Note that in
Benth~\cite{Be11}, a model for the spot given by
\eqref{Equ_OU_Brownian}-\eqref{Equ_OU_Levy} has been shown to fit gas prices
reasonably well.

We note that in electricity markets, the delivery of the underlying power
takes place over a period of time $[T_{1},T_{2}],$ where $0<T_{1}%
<T_{2}<T^{\ast}.$ We call such contracts swap contracts and we will denote
their price at time $t\leq T_{1}$ by
\[
F_{Q}(t,T_{1},T_{2})\triangleq\mathbb{E}_{Q}\left[  \frac{1}{T_{2}-T_{1}}%
\int_{T_{1}}^{T_{2}}S(T)dT|\mathcal{F}_{t}\right]  .
\]
We can use the stochastic Fubini theorem to relate the price of forwards and
swaps%
\[
F_{Q}(t,T_{1},T_{2})\triangleq\frac{1}{T_{2}-T_{1}}\int_{T_{1}}^{T_{2}}%
F_{Q}(t,T)dT.
\]
The risk premium for forward prices with a fixed delivery time is defined by
the following expression
\[
R_{Q}^{F}(t,T)\triangleq\mathbb{E}_{Q}[S(T)|\mathcal{F}_{t}]-\mathbb{E}%
_{P}[S(T)|\mathcal{F}_{t}],
\]
and for swap prices by%
\[
R_{Q}^{S}(t,T_{1},T_{2})\triangleq F_{Q}(t,T_{1},T_{2})-\mathbb{E}_{Q}%
[\frac{1}{T_{2}-T_{1}}\int_{T_{1}}^{T_{2}}S(T)dT|\mathcal{F}_{t}]
\]
It is simple to see that
\[
R_{Q}^{S}(t,T_{1},T_{2})=\frac{1}{T_{2}-T_{1}}\int_{T_{1}}^{T_{2}}R_{Q}%
^{F}(t,T)dT.
\]
The risk premium measures the price discount a producer (seller) of power must
accept compared to the predicted spot price at delivery. We shall use the risk
premium to analyse the effect of our measure change on forward prices, and to
discuss these in relation to stylized facts from the power markets.

\section{Arithmetic spot model}

We are interested in applying the previous probability measure change to study
the implied risk premium. The first model for the spot price $S$ that we are
going to consider is the arithmetic one. We define the \textit{arithmetic spot
price model} by%
\begin{equation}
S(t)=\Lambda_{a}(t)+X(t),\quad t\in\lbrack0,T^{\ast}], \label{Equ_Arith_Model}%
\end{equation}
where $T^{\ast}>0$ is a fixed time horizon. The processes $\Lambda_{a}$ is
assumed to be deterministic and it accounts for the seasonalities observed in
the spot prices. We note in passing that such models have been considered by
several authors for various energy markets. We refer to Lucia and
Schwartz~\cite{LS} for power markets and Dornier and Querel~\cite{DQ} for
temperature derivatives with no stochastic volatility. More recently Benth,
\v{S}altyt\.{e} Benth and Koekebakker~\cite{BeSa-BeKo08} has a general
discussion of arithmetic models in energy markets (see also Garcia,
Kl\"uppelberg and M\"uller~\cite{GKM} for power markets), and Benth,
\v{S}altyt\.{e} Benth \cite{BSB-sv} for temperature markets with stochastic volatility.

In order to compute the forward prices and the risk premium associated to them
in this model, we need to know the dynamics of $S$ (that is, of $X$ and
$\sigma^{2}$) under $P$ and under $Q.$ Explicit expressions for $X$ and
$\sigma^{2}$ under $P$ are given by equations $\left(  \ref{Equ_X_Explicit_P}%
\right)  $ and $\left(  \ref{Equ_Sigma_Explicit_P}\right)  ,$ respectively. In
the rest of this section, $Q=Q_{\bar{\theta},\bar{\beta}},\bar{\theta}\in
\bar{D}_{L},\bar{\beta}\in\lbrack0,1]^{2}$ defined by $\left(  \ref{EquDefQ}%
\right)  ,$ and the explicit expressions for $X$ and $\sigma^{2}$ under $Q$
are given in Remark $\ref{Remark_Dynamics_Q},$ equations $\left(
\ref{EquXDynam-Q}\right)  $ and $\left(  \ref{EquSigmaDynam-Q}\right)  ,$ respectively.

\begin{proposition}
The forward price $F_{Q}(t,T)$ in the arithmetic spot model
\ref{Equ_Arith_Model} is given by
\[
F_{Q}(t,T)=\Lambda_{a}(T)+X(t)e^{-\alpha(1-\beta_{1})(T-t)}+\frac{\theta_{1}%
}{\alpha(1-\beta_{1})}(1-e^{-\alpha(1-\beta_{1})(T-t)}).
\]

\end{proposition}

\begin{proof}
By equation $\left(  \ref{EquXDynam-Q}\right)  $ and using the basic
properties of the conditional expectation we have that%
\begin{align*}
F_{Q}(t,T) &  =\mathbb{E}_{Q}[S(T)|\mathcal{F}_{t}]=\Lambda_{a}%
(T)+X(t)e^{-\alpha(1-\beta_{1})(T-t)}\\
&  \qquad+\frac{\theta_{1}}{\alpha(1-\beta_{1})}(1-e^{-\alpha(1-\beta
_{1})(T-t)})\\
&  \qquad+\mathbb{E}_{Q}[\int_{t}^{T}\sigma(s)e^{-\alpha(1-\beta_{1}%
)(T-s)}dW_{Q}(s)|\mathcal{F}_{t}].
\end{align*}
Hence, the proof follows by showing that $\sigma(t)e^{\alpha(1-\beta_{1})t}$
belongs to $L^{2}(\Omega\times\lbrack0,T],Q\otimes dt)$because, then,
$\int_{0}^{t}\sigma(s)e^{\alpha(1-\beta_{1})s}dW_{Q}(s)$ is a $Q$-martingale
and%
\[
\mathbb{E}_{Q}[\int_{t}^{T}\sigma(s)e^{-\alpha(1-\beta_{1})(T-s)}%
dW_{Q}(s)|\mathcal{F}_{t}]=e^{-\alpha(1-\beta_{1})T}\mathbb{E}_{Q}[\int
_{t}^{T}\sigma(s)e^{\alpha(1-\beta_{1})s}dW_{Q}(s)|\mathcal{F}_{t}]=0.
\]

Using the dynamics of $\sigma^{2}$ under $Q,$ see equation
$(\ref{EquSigmaDynam-Q})$, we get%
\begin{align*}
\mathbb{E}_{Q}[\sigma^{2}(t)] &  =\sigma^{2}(0)e^{-\rho(1-\beta_{2})t}%
+\frac{\kappa_{L}^{\prime}(\theta_{2})}{\rho(1-\beta_{2})}(1-e^{-\rho
(1-\beta_{2})t})\\
&  \qquad+\mathbb{E}_{Q}[\int_{0}^{t}\int_{0}^{\infty}e^{-\rho(1-\beta
_{2})(t-s)}z\tilde{N}_{Q}^{L}(ds,dz)]\\
&  \leq\sigma^{2}(0)+\kappa_{L}^{\prime}(\theta_{2})t
\end{align*}
because $\int_{0}^{t}\int_{0}^{\infty}e^{-\rho(1-\beta_{2})s}z\tilde{N}%
_{Q}^{L}(ds,dz)$ is a $Q$-martingale starting at $0,$ see Lemma 4.3 in Benth
and Ortiz-Latorre \cite{BeO-L13}. Hence,%
\begin{align*}
\mathbb{E}_{Q}[\int_{0}^{T}\sigma^{2}(t)e^{2\alpha t}dt] &  =\int_{0}%
^{T}\mathbb{E}_{Q}[\sigma^{2}(t)]e^{2\alpha t}dt\\
&  \leq\int_{0}^{T}\left(  \sigma^{2}(0)+\kappa_{L}^{\prime}(\theta
_{2})t\right)  e^{2\alpha t}dt\\
&  \leq T\left(  \sigma^{2}(0)+\kappa_{L}^{\prime}(\theta_{2})T\right)
e^{2\alpha T}<\infty,
\end{align*}
and we can conclude.
\end{proof}

Using the previous result on forward prices we get the following formula for
the risk premium.

\begin{theorem}
\label{Theo_RiskPremiumArithm}The risk premium $R_{Q}^{F}(t,T)$ for the
forward price in the arithmetic spot model $\left(  \ref{Equ_Arith_Model}%
\right)  $ is given by%
\[
R_{Q}^{F}(t,T)=X(t)e^{-\alpha(T-t)}\left(  e^{\alpha\beta_{1}(T-t)}-1\right)
+\frac{\theta_{1}}{\alpha(1-\beta_{1})}(1-e^{-\alpha(1-\beta_{1})(T-t)}).
\]

\end{theorem}

We now analyse the risk premium in more detail under various conditions.

\subsection{Discussion on the risk premium}

The first remarkable property of this measure change is that it only depends
on the parameters that change the speed and level of mean reversion, i.e.,
$\theta_{1}$ and $\beta_{1}.$ Moreover, if $\theta_{1}=\beta_{1}=0$ we have
$R_{Q}^{F}(t,T)\equiv0,$ whatever the values of $\theta_{2}$ and $\beta_{2}$
This means that, in the arithmetic model, we can have very different pricing
measures regarding the volatility properties and have zero risk premium. In
other words, there is an unspanned volatility component that can not be
explained by just observing the forward curve. Secondly, as long as the
parameter $\beta_{1}\neq0,$ the risk premium is stochastic. Note that when
$\bar{\beta}=(0,0),$ our measure change coincides with the Esscher transform.
In the Esscher case, the risk premium has a deterministic evolution given by%
\begin{equation}
R_{Q}^{F}(t,T)=\frac{\theta_{1}}{\alpha}(1-e^{-\alpha(T-t)}),
\label{Equ_RPA_Esscher}%
\end{equation}
an already known result, see Benth \cite{Be11}.

From now on we shall rewrite the expressions for the risk premium in terms of
the time to maturity $\tau=T-t$ and, slightly abusing the notation, we will
write $R_{Q}^{F}(t,\tau)$ instead of $R_{Q}^{F}(t,t+\tau).$ We fix the
parameters of the model under the historical measure $P,$ i.e., $\alpha$ and
$\rho,$ and study the possible sign of $R_{Q}^{F}(t,\tau)$ in terms of the
change of measure parameters, i.e., $\bar{\beta}=(\beta_{1},\beta_{2})$ and
$\bar{\theta}=(\theta_{1},\theta_{2})$ and the time to maturity $\tau.$ In
fact, we just change $\theta_{1}$ and $\beta_{1}$ because the risk premium
does not depend on the values of $\theta_{2}$ and $\beta_{2}.$ Note that
present time $t$ just enters into the picture through the stochastic component
$X$ and not through the volatility process $\sigma^{2}(t).$ We are going to
study the cases $\theta_{1}=0,\beta_{1}=0$ and the general case separately.
Moreover, in order to graphically illustrate the discussion we plot the risk
premium profiles obtained assuming that the subordinator $L$ is a compound
Poisson process with jump intensity $c/\lambda>0$ and exponential jump sizes
with mean $\lambda.$ That is, $L$ will have the L\'{e}vy measure given in
Example \ref{Example_Subordinators}. We shall measure the time to maturity
$\tau$ in days and plot $R_{Q}^{F}(t,\tau)$ for $\tau\in\lbrack0,360],$
roughly one year. We fix the values of the following parameters%
\[
\alpha=0.127,\rho=1.11,c=0.4,\lambda=2.
\]
The speed of mean reversion for the factor $\alpha$ yields a half-life of
$\log(2)/0.127=5.47$ days, while the one for the volatility $\rho$ yields a
half-life of $\log(2)/1.11=0.65$ days (see e.g., Benth, \v{S}altyt\.{e} Benth
and Koekebakker \cite{BeSa-BeKo08} for the concept of half-life). The values
for $c$ and $\lambda$ give jumps with mean $0.5$ and frequency of $5$ spikes
in the volatility per month. The values for the speed of mean reversion are
obtained from an empirical analysis of the UK gas spot prices conducted in
Benth \cite{Be11}.

The following lemma will help in the discussion.

\begin{lemma}
\label{LemmaArithmetic}We have that
\begin{align*}
R_{Q}^{F}(t,\tau)  &  =X(t)e^{-\alpha\tau}\left(  e^{\alpha\beta_{1}\tau
}-1\right)  +\frac{\theta_{1}}{\alpha(1-\beta_{1})}(1-e^{-\alpha(1-\beta
_{1})\tau})\,.
\end{align*}
Moreover,
\[
\lim_{\tau\rightarrow\infty}R_{Q}^{F}(t,\tau) =\frac{\theta_{1}}%
{\alpha(1-\beta_{1})}\,,\qquad\text{and}\qquad\lim_{\tau\rightarrow0}%
\frac{\partial}{\partial\tau}R_{Q}^{F}(t,\tau) =X(t)\alpha\beta_{1}+\theta
_{1}\,.
\]

\end{lemma}

\begin{proof}
Follows easily form the expression of $R_{Q}^{F}(t,\tau)$ in Theorem
\ref{Theo_RiskPremiumArithm}.
\end{proof}

\begin{itemize}
\item \textbf{Changing the level of mean reversion (Esscher transform)}:
Setting $\beta_{1}=0,$ the probability measure $Q$ only changes the level of
mean reversion for the factor $X$ (which is assumed to be zero under the
historical measure $P$). On the other hand, the risk premium is deterministic
and cannot change with changing market conditions. From equation $\left(
\ref{LemmaArithmetic}\right)  ,$ we get that the sign of $R_{Q}^{F}(t,\tau)$
is the same for any time to maturity $\tau$ and it is equal to the sign of
$\theta_{1}.$ See Figures \ref{Essch1} and \ref{Essch2}.

\item \textbf{Changing the speed of mean reversion:} Setting $\theta_{1}=0,$
the probability measure $Q$ only changes the speed of mean reversion for the
factor $X$. Note that in this case the risk premium is stochastic and it
changes with market conditions. By Lemma $\ref{LemmaArithmetic}$ we have that
the risk premium is given by
\[
R_{Q}^{F}(t,\tau)=X(t)e^{-\alpha\tau}\left(  e^{\alpha\beta_{1}\tau}-1\right)
,
\]
with $R^{Q}(t,\tau)\rightarrow0$ as time to maturity $\tau$ tends to infinity.
On the other hand, we have that
\begin{align*}
\lim_{\tau\rightarrow0}\frac{\partial}{\partial\tau}R_{Q}^{F}(t,\tau)  &
=X(t)\alpha\beta_{1}.
\end{align*}
Hence the risk premium will vanish in the long end of the market. In the short
end, it can be both positive or negative and stochastically varying with
$X(t)$. See Figure \ref{Speed1}, where the impact of $X(t)$ in the short end
is evident as a strongly increasing (from zero) risk premium. A negative value
of $X(t)$ would lead to a downward pointing risk premium, before converging to zero.

\item \textbf{Changing the level and speed of mean reversion simultaneously}:
In the general case we can get risk premium profiles with positive values in
the short end of the forward curve and negative values in the long end, by
choosing $\theta_{1}<0$ but close to zero and $\beta_{1}$ close to $1,$
assuming that $X(t)$ is positive. See Figure \ref{GeneralArit}. We recall from
Geman \cite{Geman} that there is empirical and economical evidence for a
positive risk premium in the short end of the power forward market, while in
the long end one expects the sign of the risk premium to be negative as is the
typical situation in commodity forward markets.

\end{itemize}

%

\begin{figure}
\subfloat[][$\theta_1=0.3,\beta_1=0.0$ ]{\includegraphics
[width=7.2cm]{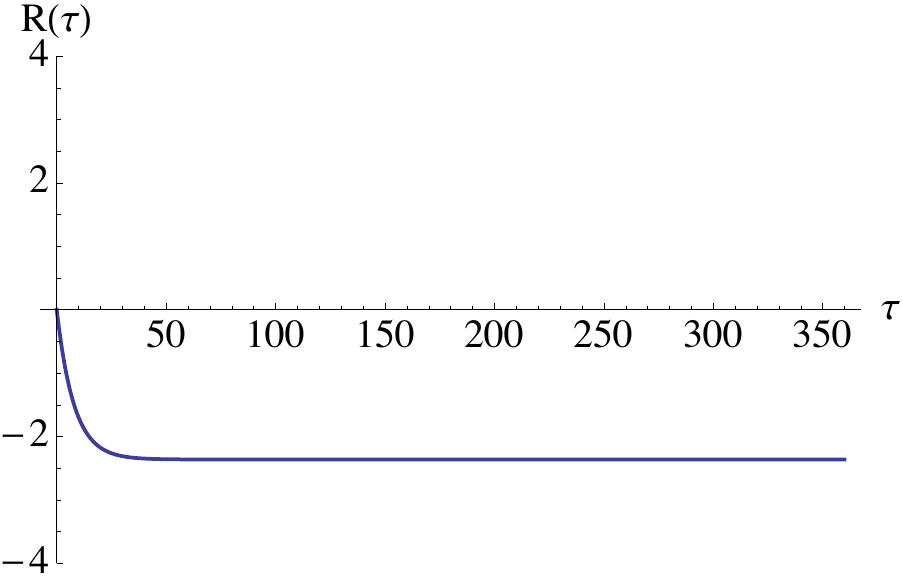}\label{Essch1}}
\qquad\subfloat[][$\theta_1=-0.3,\beta_1=0.0$]{\includegraphics
[width=7.2cm]{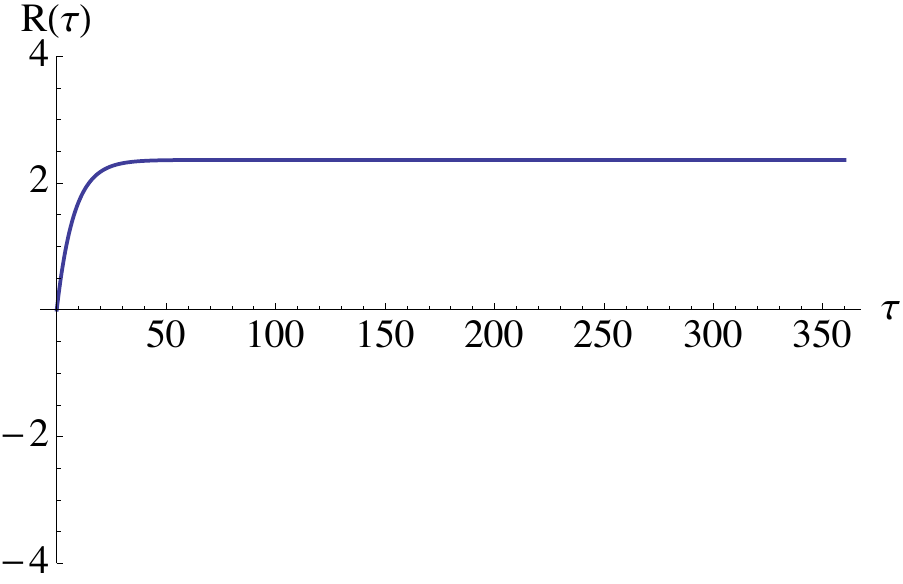}\label{Essch2}}
\\
\subfloat[][$\theta_1=0.0,\beta_1=0.9$ ]{\includegraphics
[width=7.2cm]{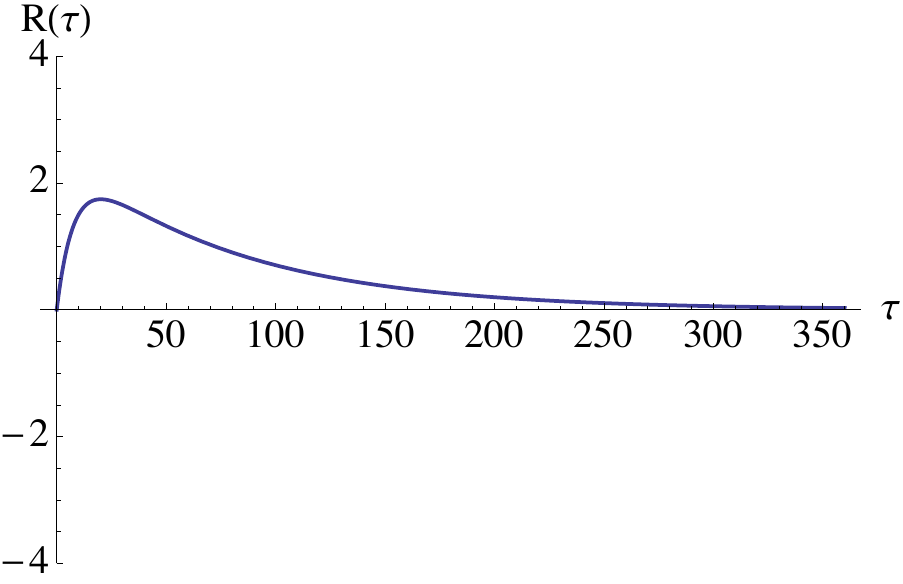}\label{Speed1}}
\qquad\subfloat[][$\theta_1=-0.04,\beta_1=0.9$]{\includegraphics
[width=7.2cm]{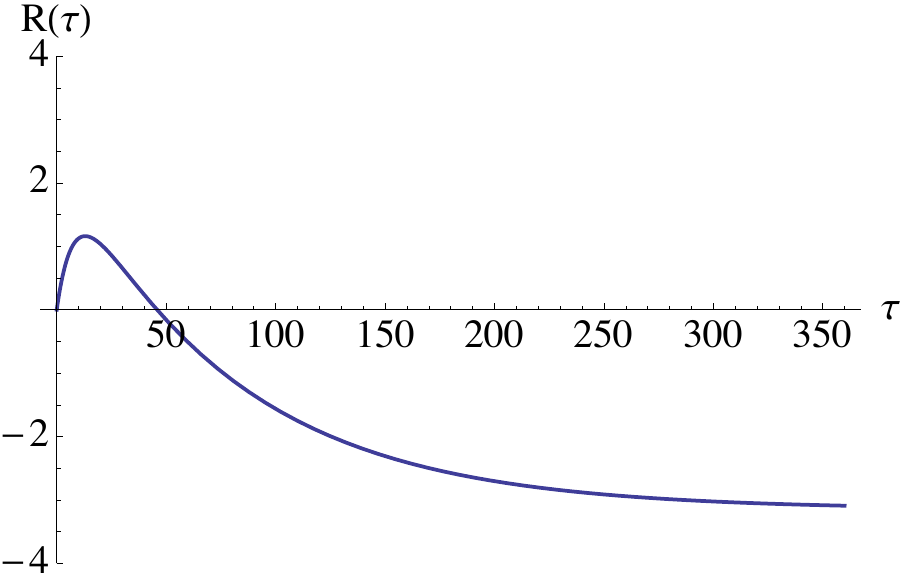}\label{GeneralArit}}
\caption
{Risk premium profiles when $L$ is a Compound Poisson process with exponentially distributed jumps.
We take $\rho=1.11,\alpha=0.127,\lambda=2,c=0.4,X(t)=2.5,\sigma(t)=0.25.$ }
\label{FigARP}
\end{figure}%

\begin{remark}
Note that in order to get a change of sign in the risk premium one must change
the level and speed of mean reversion simultaneously, see Figure
\ref{GeneralArit}. It is not possible to get the sign change by using solely
the Esscher transform, or only modifying the speed of mean reversion of the factors.
\end{remark}

\section{Geometric spot model}

The second model for the spot price $S$ is the geometric one. We define the
\textit{geometric spot price model} by
\begin{equation}
S(t)=\Lambda_{g}(t)\exp(X(t)),\quad t\in\lbrack0,T^{\ast}],
\label{Equ_Geom_Model}%
\end{equation}
where $T^{\ast}>0$ is a fixed time horizon. The process $\Lambda_{g}$ is
assumed to be deterministic and it accounts for the seasonalities observed in
the spot prices. The forward and the swap contracts are defined analogously to
the arithmetic model.

\begin{proposition}
\label{Prop_ForwardsGeometric}The forward price $F_{Q}(t,T)$ in the geometric
spot model $\left(  \ref{Equ_Geom_Model}\right)  $ is given by
\begin{align*}
F_{Q}(t,T) &  =\Lambda_{g}(T)\exp\left(  X(t)e^{-\alpha(1-\beta_{1}%
)(T-t)}+\sigma^{2}(t)e^{-\rho(1-\beta_{2})(T-t)}\frac{1-e^{-(2\alpha
-\rho(1-\beta_{2}))(T-t)}}{2(2\alpha-\rho(1-\beta_{2}))}\right)  \\
&  \qquad\times\exp\left(  \frac{\kappa_{L}^{\prime}(\theta_{2})}%
{2\rho(1-\beta_{2})}\left(  \frac{1-e^{-2\alpha(T-t)}}{2\alpha}-e^{-\rho
(1-\beta_{2})(T-t)}\frac{1-e^{-(2\alpha-\rho(1-\beta_{2}))(T-t)}}%
{(2\alpha-\rho(1-\beta_{2}))}\right)  \right)  \\
&  \qquad\times\exp\left(  \frac{\theta_{1}}{\alpha(1-\beta_{1})}%
(1-e^{-\alpha(1-\beta_{1})(T-t)})\right)  \\
&  \qquad\times\mathbb{E}_{Q}\left[  \exp\left(  \frac{e^{-2\alpha T}}{2}%
\int_{t}^{T}e^{(2\alpha-\rho(1-\beta_{2}))s}\left(  \int_{t}^{s}\int
_{0}^{\infty}e^{\rho(1-\beta_{2})u}z\tilde{N}_{Q}^{L}(du,dz)\right)
ds\right)  |\mathcal{F}_{t}\right]
\end{align*}
In the particular case $Q=P,$ it holds that%
\begin{align*}
F_{P}(t,T) &  =\Lambda_{g}(T)\exp\left(  X(t)e^{-\alpha(T-t)}+\sigma
^{2}(t)e^{-\rho(T-t)}\frac{1-e^{-(2\alpha-\rho)(T-t)}}{2(2\alpha-\rho
)}\right)  \\
&  \qquad\times\exp\left(  \int_{0}^{T-t}\kappa_{L}\left(  e^{-\rho s}%
\frac{1-e^{-(2\alpha-\rho)s}}{2(2\alpha-\rho)}\right)  ds\right)  .
\end{align*}

\end{proposition}

\begin{proof}
Denote by $\mathcal{F}_{t}^{L}$ the sigma algebra generated by the process $L$
up to time $t.$ Then, we have that
\begin{align*}
\mathbb{E}_{Q}[S(T)|\mathcal{F}_{t}]  &  =\Lambda_{g}(T)\mathbb{E}_{Q}%
[\exp(X(T))|\mathcal{F}_{t}]\\
&  =\Lambda_{g}(T)\exp\left(  X(t)e^{-\alpha(1-\beta_{1})(T-t)}+\frac
{\theta_{1}}{\alpha(1-\beta_{1})}(1-e^{-\alpha(1-\beta_{1})(T-t)}\right) \\
&  \qquad\times\mathbb{E}_{Q}\left[  \exp\left(  \int_{t}^{T}\sigma
(s)e^{-\alpha(1-\beta_{1})(T-s)}dW_{Q}(s)\right)  |\mathcal{F}_{t}\right] \\
&  =\Lambda_{g}(T)\exp\left(  X(t)e^{-\alpha(1-\beta_{1})(T-t)}+\frac
{\theta_{1}}{\alpha(1-\beta_{1})}(1-e^{-\alpha(1-\beta_{1})(T-t)}\right) \\
&  \qquad\times\mathbb{E}_{Q}\left[  \mathbb{E}_{Q}\left[  \exp\left(
\int_{t}^{T}\sigma(s)e^{-\alpha(T-s)}dW_{Q}(s)\right)  |\mathcal{F}_{T}%
^{L}\vee\mathcal{F}_{t}\right]  |\mathcal{F}_{t}\right] \\
&  =\Lambda_{g}(T)\exp\left(  X(t)e^{-\alpha(1-\beta_{1})(T-t)}+\frac
{\theta_{1}}{\alpha(1-\beta_{1})}(1-e^{-\alpha(1-\beta_{1})(T-t)}\right) \\
&  \qquad\times\mathbb{E}_{Q}\left[  \exp\left(  \frac{1}{2}\int_{t}^{T}%
\sigma^{2}(s)e^{-2\alpha(T-s)}ds\right)  |\mathcal{F}_{t}\right]  .
\end{align*}
On the other hand, the dynamics of $\sigma^{2}(s)$ can be written, for $s>t,$
as%
\begin{align*}
\sigma^{2}(s)  &  =\sigma^{2}(t)e^{-\rho(1-\beta_{2})(s-t)}+\frac{\kappa
_{L}^{\prime}(\theta_{2})}{\rho(1-\beta_{2})}(1-e^{-\rho(1-\beta_{2})(s-t)})\\
&  \qquad+\int_{t}^{s}\int_{0}^{\infty}e^{-\rho(1-\beta_{2})(s-u)}z\tilde
{N}_{Q}^{L}(du,dz).
\end{align*}
Then, we get%
\begin{align*}
\mathbb{E}_{Q}[S(T)|\mathcal{F}_{t}]  &  =\Lambda_{g}(T)\exp\left(
X(t)e^{-\alpha(1-\beta_{1})(T-t)}+\frac{\theta_{1}}{\alpha(1-\beta_{1}%
)}(1-e^{-\alpha(1-\beta_{1})(T-t)}\right) \\
&  \quad\times\mathbb{E}_{Q}\left[  \exp\left(  \frac{1}{2}\int_{t}%
^{T}\{\sigma^{2}(t)e^{-\rho(1-\beta_{2})(s-t)}+\frac{\kappa_{L}^{\prime
}(\theta_{2})}{\rho(1-\beta_{2})}(1-e^{-\rho(1-\beta_{2})(s-t)})\right.
\right. \\
&  \quad\quad+\left.  \left.  \int_{t}^{s}\int_{0}^{\infty}e^{-\rho
(1-\beta_{2})(s-u)}zN_{Q}^{L}(du,dz)\}e^{-2\alpha(T-s)}ds\right)
|\mathcal{F}_{t}\right] \\
&  =\Lambda_{g}(T)\exp\left(  X(t)e^{-\alpha(1-\beta_{1})(T-t)}+\sigma
^{2}(t)e^{-\rho(1-\beta_{2})(T-t)}\frac{1-e^{-(2\alpha-\rho(1-\beta
_{2}))(T-t)}}{2(2\alpha-\rho(1-\beta_{2}))}\right) \\
&  \quad\times\exp\left(  \frac{\kappa_{L}^{\prime}(\theta_{2})}{2\rho
(1-\beta_{2})}\left(  \frac{1-e^{-2\alpha(T-t)}}{2\alpha}-e^{-\rho(1-\beta
_{2})(T-t)}\frac{1-e^{-(2\alpha-\rho(1-\beta_{2}))(T-t)}}{(2\alpha
-\rho(1-\beta_{2}))}\right)  \right) \\
&  \quad\times\exp\left(  \frac{\theta_{1}}{\alpha(1-\beta_{1})}%
(1-e^{-\alpha(1-\beta_{1})(T-t)})\right) \\
&  \quad\times\mathbb{E}_{Q}\left[  \exp\left(  \frac{e^{-2\alpha T}}{2}%
\int_{t}^{T}e^{(2\alpha-\rho(1-\beta_{2}))s}\left(  \int_{t}^{s}\int
_{0}^{\infty}e^{\rho(1-\beta_{2})u}z\tilde{N}_{Q}^{L}(du,dz)\right)
ds\right)  |\mathcal{F}_{t}\right]
\end{align*}
Now, taking into account that $N_{Q}^{L}(du,dz)$ has independent increments
for $Q=P,$ using a stochastic version of Fubini's Theorem and the exponential
moments formula for Poisson random measures we obtain
\begin{align*}
\mathbb{E}_{P}  &  \left[  \exp\left(  \frac{e^{-2\alpha T}}{2}\int_{t}%
^{T}e^{(2\alpha-\rho)s}\left(  \int_{t}^{s}\int_{0}^{\infty}e^{\rho u}%
z\tilde{N}^{L}(du,dz)\right)  ds\right)  |\mathcal{F}_{t}\right] \\
&  \qquad=\mathbb{E}_{P}\left[  \exp\left(  \frac{e^{-2\alpha T}}{2}\int
_{t}^{T}e^{(2\alpha-\rho)s}\left(  \int_{t}^{s}\int_{0}^{\infty}e^{\rho
u}zN^{L}(du,dz)\right)  ds\right)  \right] \\
&  \qquad\qquad\times\exp\left(  -\frac{e^{-2\alpha T}}{2}\int_{t}%
^{T}e^{(2\alpha-\rho)s}\left(  \int_{t}^{s}\int_{0}^{\infty}e^{\rho u}%
z\ell(dz)du\right)  ds\right) \\
&  \qquad=\mathbb{E}_{P}\left[  \exp\left(  \int_{t}^{T}\int_{0}^{\infty
}e^{-\rho(T-u)}\frac{1-e^{-(2\alpha-\rho(T-u)}}{2(2\alpha-\rho)}%
zN^{L}(du,dz)\right)  \right] \\
&  \qquad\qquad\times\exp\left(  -\frac{e^{-2\alpha T}\kappa_{L}^{\prime}%
(0)}{2\rho}\int_{t}^{T}e^{(2\alpha-\rho)s}\left(  e^{\rho s}-e^{\rho
t}\right)  ds\right) \\
&  \qquad=\exp\left(  \int_{t}^{T}\int_{0}^{\infty}\left(  \exp\left(
e^{-\rho(T-u)}\frac{1-e^{-(2\alpha-\rho)(T-u)}}{2(2\alpha-\rho)}z\right)
-1\right)  \ell(dz)du\right) \\
&  \qquad\qquad\times\exp\left(  -\frac{e^{-2\alpha T}\kappa_{L}^{\prime}%
(0)}{2\rho}\left(  \frac{e^{2\alpha T}-e^{2\alpha t}}{2\alpha}-\frac
{e^{2\alpha T}e^{-\rho(T-t)}-e^{2\alpha t}}{2\alpha-\rho}\right)  \right) \\
&  \qquad=\exp\left(  \int_{0}^{T-t}\kappa_{L}\left(  e^{-\rho s}%
\frac{1-e^{-(2\alpha-\rho)s}}{2(2\alpha-\rho)}\right)  ds\right) \\
&  \qquad\qquad\times\exp\left(  -\frac{\kappa_{L}^{\prime}(0)}{2\rho}\left(
\frac{1-e^{-2\alpha(T-t)}}{2\alpha}+\frac{e^{-2\alpha(T-t)}-e^{-\rho(T-t)}%
}{2\alpha-\rho}\right)  \right)  .
\end{align*}
In the last equality we have used the definition of $\kappa_{L}(\theta)$ and
the change of variable $s=T-u.$ Finally, combining the previous expression
with the expression for $\mathbb{E}_{Q}[S(T)|\mathcal{F}_{t}]$ with $Q=P,$
i.e., with $\beta_{1}=\beta_{2}=\theta_{1}=\theta_{2}=0$ we get the result
\end{proof}

\begin{remark}
Note that $\mathbb{E}_{Q}[S(T)]$ can be infinite. In the case $Q=P,$ if
$\Theta_{L}=\infty,$ then $\mathbb{E}_{P}[S(T)]<\infty.$ However, if
$\Theta_{L}<\infty,$ then $\mathbb{E}_{P}[S(T)]<\infty$ if and only if%
\begin{equation}
\int_{0}^{T}\kappa_{L}\left(  e^{-\rho s}\frac{1-e^{-(2\alpha-\rho)s}%
}{2(2\alpha-\rho)}\right)  ds<\infty. \label{CondFiniteUnderP}%
\end{equation}

\end{remark}

Condition $\left(  \ref{CondFiniteUnderP}\right)  $ imposes some restrictions
on the parameters of the model $\alpha,\rho$ and $\Theta_{L}.$ For
$\alpha,\rho>0,$ consider the function $\Upsilon_{\alpha,\rho}(t)=e^{-\rho
t}\frac{1-e^{-(2\alpha-\rho)t}}{2(2\alpha-\rho)},t>0.$ It is easy to see that
this function is strictly positive and achieves its maximum at $t^{\ast}%
=-\log(\rho/2\alpha)/(2\alpha-\rho)$ with value $\Upsilon_{\alpha,\rho
}(t^{\ast})=\frac{1}{2\rho}\left(  \frac{\rho}{2\alpha}\right)  ^{\frac
{1}{1-\frac{\rho}{2\alpha}}}.$ Then, it is natural to impose the following
assumption on the model parameter that guarantees that condition $\left(
\ref{CondFiniteUnderP}\right)  $ is satisfied for all $T>0$:

\begin{assumption}
[$\mathcal{P}$]We assume that $\alpha,\rho>0$ and $\Theta_{L}$ satisfy
\[
\frac{1}{2\rho}\left(  \frac{\rho}{2\alpha}\right)  ^{\frac{1}{1-\frac{\rho
}{2\alpha}}}\leq\Theta_{L}-\delta,
\]
for some $\delta>0.$
\end{assumption}

Obviously, if $\Theta_{L}=\infty$ then assumption $\mathcal{P}$ is satisfied.
Suppose that $\Theta_{L}<\infty,$ then if we choose $\rho$ close to zero the
value of $\alpha$ must be bounded away from zero, and viceversa, for
assumption $\mathcal{P}$ to be satisfied.

The risk premium in the geometric case becomes:

\begin{theorem}
The risk premium $R_{Q}^{F}(t,T)$ for the forward price in the geometric spot
model $\left(  \ref{Equ_Geom_Model}\right)  $ is given by%
\begin{align*}
R_{Q}^{F}(t,T) &  =\mathbb{E}_{P}[S(T)|\mathcal{F}_{t}]\left\{  \exp\left(
X(t)e^{-\alpha(T-t)}(e^{\alpha\beta_{1}(T-t)}-1)\right)  \right.  \\
&  \qquad\times\left.  \exp\left(  \sigma^{2}(t)\left(  e^{-\rho(1-\beta
_{2})(T-t)}\frac{1-e^{-(2\alpha-\rho(1-\beta_{2}))(T-t)}}{2(2\alpha
-\rho(1-\beta_{2}))}-e^{-\rho(T-t)}\frac{1-e^{-(2\alpha-\rho)(T-t)}}%
{2(2\alpha-\rho)}\right)  \right)  \right.  \\
&  \qquad\times\exp\left(  \frac{\kappa_{L}^{\prime}(\theta_{2})}%
{2\rho(1-\beta_{2})}\left(  \frac{1-e^{-2\alpha(T-t)}}{2\alpha}-e^{-\rho
(1-\beta_{2})(T-t)}\frac{1-e^{-(2\alpha-\rho(1-\beta_{2}))(T-t)}}%
{(2\alpha-\rho(1-\beta_{2}))}\right)  \right)  \\
&  \qquad\times\exp\left(  \frac{\theta_{1}}{\alpha(1-\beta_{1})}%
(1-e^{-\alpha(1-\beta_{1})(T-t)})\right)  \\
&  \qquad\times\exp\left(  -\int_{0}^{T-t}\kappa_{L}\left(  e^{-\rho s}%
\frac{1-e^{-(2\alpha-\rho)s}}{2(2\alpha-\rho)}\right)  ds\right)  \\
&  \qquad\times\left.  \mathbb{E}_{Q}\left[  \exp\left(  \frac{e^{-2\alpha T}%
}{2}\int_{t}^{T}e^{(2\alpha-\rho)s}\left(  \int_{t}^{s}\int_{0}^{\infty
}e^{\rho u}zN_{Q}^{L}(du,dz)\right)  ds\right)  |\mathcal{F}_{t}\right]
-1\right\}
\end{align*}

\end{theorem}

\begin{proof}
This follows immediately from Proposition~\ref{Prop_ForwardsGeometric}.
\end{proof}

The risk premium in the geometric case becomes hard to analyse due to the
presence of the conditional expectation in the last term, involving the jump
process $N_{Q}^{L}$ with respect to $Q$. In the remainder of this Section we
shall rather exploit the affine structure of the model to analyse the risk premium.

\subsection{An analysis of the risk premium based on the affine structure}

An alternative way of computing $\mathbb{E}_{Q}[S(T)|\mathcal{F}_{t}],$ which
can provide semi-explicit expressions, is to use the affine structure of
$Z=(Z_{1}(t),Z_{2}(t))^{\top}=(\sigma^{2}(t),X(t))^{\top}.$ Let $\Lambda
_{i}^{\bar{\theta},\bar{\beta}}(u),i=0,1,2,$ be the L\'{e}vy exponents
associated to the affine characteristics in Remark
\ref{RemarkAffineCharactUnderQ}, i.e.,%
\begin{align*}
\Lambda_{0}^{\bar{\theta},\bar{\beta}}(u_{1},u_{2})  &  =\beta_{0}^{\top
}u+\frac{1}{2}u^{\top}\gamma_{0}u+\int(e^{u_{1}z_{1}+u_{2}z_{2}}-1-u_{1}%
z_{1}-u_{2}z_{2})\varphi_{0}(dz)\\
&  =\kappa_{L}^{\prime}(\theta_{2})u_{1}+\theta_{1}u_{2}+\int_{0}^{\infty
}(e^{u_{1}z_{1}}-1-u_{1}z_{1})e^{\theta_{2}z_{1}}\ell(dz_{1})\\
&  =\theta_{1}u_{2}+\kappa_{L}(u_{1}+\theta_{2})-\kappa_{L}(\theta_{2}),\\
\Lambda_{1}^{\bar{\theta},\bar{\beta}}(u_{1},u_{2})  &  =\beta_{1}^{\top
}u+u^{\top}\gamma_{1}u+\int(e^{u_{1}z_{1}+u_{2}z_{2}}-1-u_{1}z_{1}-u_{2}%
z_{2})\varphi_{1}(dz)\\
&  =-\rho(1-\beta_{2})u_{1}+\frac{u_{2}^{2}}{2}+\frac{\rho\beta_{2}}%
{\kappa_{L}^{\prime\prime}(\theta_{2})}\int_{0}^{\infty}(e^{u_{1}z_{1}%
}-1-u_{1}z_{1})z_{1}e^{\theta_{2}z_{1}}\ell(dz_{1})\\
&  =-\rho u_{1}+\frac{u_{2}^{2}}{2}+\frac{\rho\beta_{2}}{\kappa_{L}%
^{\prime\prime}(\theta_{2})}(\kappa_{L}^{\prime}(u_{1}+\theta_{2})-\kappa
_{L}^{\prime}(\theta_{2})),\\
\Lambda_{2}^{\bar{\theta},\bar{\beta}}(u_{1},u_{2})  &  =\beta_{2}^{\top
}u+u^{\top}\gamma_{2}u+\int(e^{u_{1}z_{1}+u_{2}z_{2}}-1-u_{1}z_{1}-u_{2}%
z_{2})\varphi_{2}(dz)\\
&  =-\alpha(1-\beta_{1})u_{2}.
\end{align*}
We find the following:

\begin{theorem}
\label{TheoRicattiAbstractQGeneral}Let $\bar{\beta}=(\beta_{1},\beta_{2}%
)\in\lbrack0,1]^{2},\bar{\theta}=(\theta_{1},\theta_{2})\in\bar{D}_{L}.$
Assume that there exist functions $\Psi_{i}^{\bar{\theta},\bar{\beta}%
},i=0,1,2$ belonging to $C^{1}([0,T];\mathbb{R}^{2})$ satisfying the
generalised Riccati equation%
\begin{equation}%
\begin{array}
[c]{lcc}%
\frac{d}{dt}\Psi_{1}^{\bar{\theta},\bar{\beta}}(t)=-\rho\Psi_{1}^{\bar{\theta
},\bar{\beta}}(t)+\frac{(\Psi_{2}^{\bar{\theta},\bar{\beta}}(t))^{2}}{2}%
+\frac{\rho\beta_{2}}{\kappa_{L}^{\prime\prime}(\theta_{2})}(\kappa
_{L}^{\prime}(\Psi_{1}^{\bar{\theta},\bar{\beta}}(t)+\theta_{2})-\kappa
_{L}^{\prime}(\theta_{2})), &  & \Psi_{1}^{\bar{\theta},\bar{\beta}}(0)=0,\\
\frac{d}{dt}\Psi_{2}^{\bar{\theta},\bar{\beta}}(t)=-\alpha(1-\beta_{1}%
)\Psi_{2}^{\bar{\theta},\bar{\beta}}(t), &  & \Psi_{2}^{\bar{\theta}%
,\bar{\beta}}(0)=1,\\
\frac{d}{dt}\Psi_{0}^{\bar{\theta},\bar{\beta}}(t)=\theta_{1}\Psi_{2}%
^{\bar{\theta},\bar{\beta}}(t)+\kappa_{L}(\Psi_{1}^{\bar{\theta},\bar{\beta}%
}(t)+\theta_{2})-\kappa_{L}(\theta_{2}), &  & \Psi_{0}^{\bar{\theta}%
,\bar{\beta}}(0)=0,
\end{array}
\label{EquRiccatiODEGeneral}%
\end{equation}
and the integrability condition%
\begin{equation}
\sup_{t\in\lbrack0,T]}\kappa_{L}^{\prime\prime}(\theta_{2}+\Psi_{1}%
^{\bar{\theta},\bar{\beta}}(t))<\infty.\label{Equ_Integrability_ODE}%
\end{equation}
Then,
\[
\mathbb{E}_{Q}[\exp(X(T))|\mathcal{F}_{t}]=\exp\left(  \Psi_{0}^{\bar{\theta
},\bar{\beta}}(T-t)+\Psi_{1}^{\bar{\theta},\bar{\beta}}(T-t)\sigma^{2}%
(t)+\Psi_{2}^{\bar{\theta},\bar{\beta}}(T-t)X(t)\right)  ,
\]
and%
\begin{align}
R_{Q}^{F}(t,T) &  =\mathbb{E}_{P}[S(T)|\mathcal{F}_{t}]\label{EquRiskPremium}%
\\
&  \qquad\times\left\{  \exp\left(  \Psi_{0}^{\bar{\theta},\bar{\beta}%
}(T-t)-\int_{0}^{T-t}\kappa_{L}\left(  e^{-\rho s}\frac{1-e^{-(2\alpha-\rho
)s}}{2(2\alpha-\rho)}\right)  ds\right.  \right.  \nonumber\\
&  \qquad+\left(  \Psi_{1}^{\bar{\theta},\bar{\beta}}(T-t)-e^{-\rho(T-t)}%
\frac{1-e^{-(2\alpha-\rho)(T-t)}}{2(2\alpha-\rho)}\right)  \sigma
^{2}(t)\nonumber\\
&  \qquad+\left.  \left.  \left(  \Psi_{2}^{\bar{\theta},\bar{\beta}%
}(T-t)-e^{-\alpha(T-t)}\right)  X(t)\right)  -1\right\}  .\nonumber
\end{align}

\end{theorem}

\begin{proof}
The result is a consequence of Theorem 5.1 in Kallsen and
Muhle-Karbe~\cite{KaMu-ka10}: Making the change of variable $t\rightarrow
T-t$, the ODE \eqref{EquRiccatiODEGeneral} is reduced to the one appearing in
items 2. and 3. of Theorem 5.1 in Kallsen and Muhle-Karbe~\cite{KaMu-ka10}.
The integrability assumption \eqref{Equ_Integrability_ODE} implies conditions
1. and 5., in Theorem 5.1, and condition 4. in that same Theorem is trivially
satisfied because $\sigma^{2}(0)$ and $X(0)$ are deterministic. Hence, the
conclusion of Theorem 5.1 in Kallsen and Muhle-Karbe~\cite{KaMu-ka10}, with
$p=(0,1),$ holds and we get%
\begin{equation}
\mathbb{E}_{Q}[\exp(X(T))|\mathcal{F}_{t}]=\exp\left(  \Psi_{0}^{\bar{\theta
},\bar{\beta}}(T-t)+\Psi_{1}^{\bar{\theta},\bar{\beta}}(T-t)\sigma^{2}%
(t)+\Psi_{2}^{\bar{\theta},\bar{\beta}}(T-t)X(t)\right)  ,\quad t\in
\lbrack0,T]. \label{Equ_Exp_Affine_Formula}%
\end{equation}
The result on the risk premium now follows easily.
\end{proof}

A couple of remarks are in place.

\begin{remark}
The applicability of Theorem \ref{TheoRicattiAbstractQGeneral} is quite
limited as it is stated. This is due to the fact that it is very difficult to
see a priori if there exist functions $\Psi_{i}^{\bar{\theta},\bar{\beta}%
},i=0,1,2$ belonging to $C^{1}([0,T];\mathbb{R}^{2})$ satisfying equation
$\left(  \ref{EquRiccatiODEGeneral}\right)  .$ One has to study existence and
uniqueness of solutions of equation $\left(  \ref{EquRiccatiODEGeneral}%
\right)  $ and the possibility of extending the solution to arbitrary large
$T>0.$ We study this problem in Theorem \ref{TheoRicattiParticularQGeneral}.
\end{remark}

\begin{remark}
\label{RemarkPropertiesSystemODEsQGeneral}Note that the previous system of
autonomous ODEs can be effectively reduced to a one dimensional non autonomous
ODE. We have that for any $\bar{\theta}\in\bar{D}_{L},\bar{\beta}\in
\lbrack0,1]^{2}$, the solution of the second equation is given by $\Psi
_{2}^{\bar{\theta},\bar{\beta}}(t)=\exp(-\alpha(1-\beta_{1})t)$.
Plugging this solution to the first equation we get the following equation to
solve for $\Psi_{1}^{\bar{\theta},\bar{\beta}}(t)$%
\begin{equation}
\frac{d}{dt}\Psi_{1}^{\bar{\theta},\bar{\beta}}(t)=-\rho\Psi_{1}^{\bar{\theta
},\bar{\beta}}(t)+\frac{e^{-2\alpha(1-\beta_{1})t}}{2}+\frac{\rho\beta_{2}%
}{\kappa_{L}^{\prime\prime}(\theta_{2})}(\kappa_{L}^{\prime}(\Psi_{1}%
^{\bar{\theta},\bar{\beta}}(t)+\theta_{2})-\kappa_{L}^{\prime}(\theta_{2})),
\label{Equ_1D_ODE}%
\end{equation}
with initial condition $\Psi_{1}^{\bar{\theta},\bar{\beta}}(0)=0.$ The
equation for $\Psi_{0}^{\bar{\theta},\bar{\beta}}(t)$ is solved by integrating
$\Lambda_{0}^{\bar{\theta},\bar{\beta}}(\Psi_{1}^{\bar{\theta},\bar{\beta}%
}(t),\Psi_{2}^{\bar{\theta},\bar{\beta}}(t)),$ i.e.,%
\begin{align*}
\Psi_{0}^{\bar{\theta},\bar{\beta}}(t)  &  =\int_{0}^{t}\{\theta_{1}\Psi
_{2}^{\bar{\theta},\bar{\beta}}(s)+\kappa_{L}(\Psi_{1}^{\bar{\theta}%
,\bar{\beta}}(s)+\theta_{2})-\kappa_{L}(\theta_{2})\}ds\\
&  =\theta_{1}\frac{1-e^{-\alpha(1-\beta_{1})t}}{\alpha(1-\beta_{1})}+\int
_{0}^{t}\{\kappa_{L}(\Psi_{1}^{\bar{\theta},\bar{\beta}}(s)+\theta_{2}%
)-\kappa_{L}(\theta_{2})\}ds.
\end{align*}

\end{remark}

As we have already indicated, we cannot in general find the explicit solution
of the system of ODEs in Theorem~\ref{TheoRicattiAbstractQGeneral}, and has to
rely on numerical techniques. However, the main problem is to ensure the
existence and uniqueness of global solutions. Before stating our main result
on this question, we introduce some notation and a technical lemma.

\begin{lemma}
\label{LemmaLambda(a)}Let $\Lambda^{\theta,\beta,a}:[0,\Theta_{L}%
-\theta)\rightarrow\mathbb{R}$ be the function defined by%
\begin{equation}
\Lambda^{\theta,\beta,a}(u)=-\rho u+a+\frac{\rho\beta}{\kappa_{L}%
^{\prime\prime}(\theta)}(\kappa_{L}^{\prime}(u+\theta)-\kappa_{L}^{\prime
}(\theta)), \label{EquDefLambda_a}%
\end{equation}
where $a\geq0,(\theta,\beta)\in D_{L}\times(0,1)$ and consider the set%
\[
\mathcal{D}_{b}(a)=\{(\theta,\beta)\in D_{L}\times(0,1):\exists u\in
\lbrack0,\Theta_{L}-\theta)\quad s.t.\quad\Lambda^{\theta,\beta,a}(u)\leq0\}.
\]
Then, we have that:

\begin{enumerate}
\item For any $(\theta,\beta)\in D_{L}\times(0,1),$ there exists a unique
global minimum of the function $\Lambda^{\theta,\beta,a}(u)$ which is attained
at%
\begin{equation}
u^{m}(\theta,\beta)=\left(  \kappa_{L}^{\prime\prime}\right)  ^{-1}\left(
\frac{\kappa_{L}^{\prime\prime}(\theta)}{\beta}\right)  -\theta,\label{Equ_um}%
\end{equation}
with value%
\begin{align}
\Lambda^{\theta,\beta,a}(u^{m}(\theta,\beta)) &  =-\rho\left(  \left(
\kappa_{L}^{\prime\prime}\right)  ^{-1}\left(  \frac{\kappa_{L}^{\prime\prime
}(\theta_{2})}{\beta_{2}}\right)  -\theta_{2}\right)  +a\label{Equ_Lambda(um)}%
\\
&  \qquad+\left(  \frac{\rho\beta_{2}}{\kappa_{L}^{\prime\prime}(\theta_{2}%
)}(\kappa_{L}^{\prime}(\left(  \kappa_{L}^{\prime\prime}\right)  ^{-1}\left(
\frac{\kappa_{L}^{\prime\prime}(\theta_{2})}{\beta_{2}}\right)  )-\kappa
_{L}^{\prime}(\theta_{2}))\right)  .\nonumber
\end{align}

\item The function $\Lambda^{\theta,\beta,a}(u)$ is strictly decreasing in
$(0,u^{m}(\theta,\beta))$ and strictly increasing in $(u^{m}(\theta
,\beta),\Theta_{L}-\theta).$

\item For $\theta\in D_{L}$ fixed, one has that $u^{m}(\theta,\beta
)\uparrow\Theta_{L}-\theta$ when $\beta\downarrow0$ and $u^{m}(\theta
,\beta)\downarrow0$ when $\beta\uparrow1.$

\item The set $\mathcal{D}_{b}(a)$ coincides with the set
\[
\{(\theta,\beta)\in D_{L}\times(0,1):\Lambda^{\theta,\beta,a}(u^{m}%
(\theta,\beta))\leq0\}.
\]
Moreover, for $a>0,$ we have the following:

\begin{enumerate}
\item If $\theta\in D_{L}$ is such that $\theta>\Theta_{L}-a/\rho$ then
$\nexists\beta\in(0,1)$ such that $(\theta,\beta)\in\mathcal{D}_{b}(a).$

\item If $\theta\in D_{L}$ is such that $\theta<\Theta_{L}-a/\rho$ then there
exists a unique $0<\beta_{m}<1$ such that%
\begin{equation}
\Lambda^{\theta,\beta,a}(u^{m}(\theta,\beta_{m}))=0, \label{Equ_Beta_m}%
\end{equation}
and for all $\beta\in\lbrack0,\beta_{m}]$ one has $(\theta,\beta
)\in\mathcal{D}_{b}(a).$
\end{enumerate}

\item For $(\theta,\beta)\in\mathcal{D}_{b}(a),$ $a>0$ there exists a unique
zero of $\Lambda^{\theta,\beta,a}(u),$ denoted by $u_{a}^{0}(\theta,\beta).$
As a function of $\beta,$ $u_{a}^{0}(\theta,\beta)$ is well defined\ on
$[0,\beta_{m}],$ strictly increasing, with $u_{a}^{0}(\theta,0)=a/\rho$ and
$u_{a}^{0}(\theta,\beta_{m})=u^{m}(\theta,\beta_{m}).$
\end{enumerate}
\end{lemma}

\begin{proof}
\textit{Proof of }$1.$: According to Remark \ref{Remark_D_Cumulants} we have
that%
\begin{align}
\frac{d}{du}\Lambda^{\theta,\beta,a}(u)  &  =-\rho+\rho\beta\frac{\kappa
_{L}^{\prime\prime}(u+\theta)}{\kappa_{L}^{\prime\prime}(\theta)},\nonumber\\
\frac{d^{2}}{du^{2}}\Lambda^{\theta,\beta,a}(u)  &  =\rho\beta\frac{\kappa
_{L}^{\prime\prime\prime}(u+\theta)}{\kappa_{L}^{\prime\prime}(\theta
)}>0,\nonumber
\end{align}
which yields that there exists a unique $0<u^{\ast}(\theta,\beta)<\Theta
_{L}-\theta$ for $\theta\in D_{L}$ and $\beta\in(0,1)$ such that
$\Lambda^{\theta,\beta,a}(u)$ attains a global minimum. In fact, $u^{\ast
}(\theta_{2},\beta_{2})$ solves%
\[
1=\beta\frac{\kappa_{L}^{\prime\prime}(u+\theta)}{\kappa_{L}^{\prime\prime
}(\theta)}.
\]
Moreover, by Remark \ref{Remark_D_Cumulants} again, we have that $\kappa
_{L}^{\prime\prime}(u)$ is a strictly increasing function and, hence, it has a
well defined inverse $\left(  \kappa_{L}^{\prime\prime}\right)  ^{-1}(v)$
which yields that $u^{m}(\theta,\beta)$ and $\Lambda^{\theta,\beta,a}%
(u^{m}(\theta,\beta))$ are given by equations $\left(  \ref{Equ_um}\right)  $
and $\left(  \ref{Equ_Lambda(um)}\right)  ,$ respectively.

\textit{Proof of} $2.:$ It follows from the fact that $\frac{d}{du}%
\Lambda^{\theta,\beta,a}(u)<0,u\in(0,u^{m}(\theta_{2},\beta_{2}))$ and
$\frac{d}{du}\Lambda^{\theta,\beta,a}(u)>0,(u^{m}(\theta,\beta),\Theta
_{L}-\theta)$.

\textit{Proof of} $3.:$ It follows from the monotonicity of $\kappa
_{L}^{\prime\prime}(u)$ and the explicit expression of $u^{m}(\theta,\beta)$
given by equation $\left(  \ref{Equ_um}\right)  .$ Note that, as a function of
$\beta,$ $u^{m}(\theta,\beta)$ is a strictly decreasing, continuous function
in $(0,1).$

\textit{Proof of} $4.$ and $5.:$
For any $\theta\in D_{L},$ note that%
\[
a-\rho u\triangleq\Lambda^{\theta,0,a}(u)\leq\Lambda^{\theta,\beta,a}%
(u),\quad\beta\in(0,1),u\in\lbrack0,\Theta_{L}-\theta),
\]
and $\Lambda^{\theta,0,a}(u)>0$ if $u\in(0,a/\rho).$ Therefore, if
$\theta>\Theta_{L}-a/\rho$ then $\Lambda^{\theta,\beta,a}(u)>0,u\in
\lbrack0,\Theta_{L}-\theta)$ for any $\beta\in(0,1).$ On the other hand, if
$\theta<\Theta_{L}-a/\rho$ we have that $a/\rho\in\lbrack0,\Theta_{L}%
-\theta).$ Moreover, defining the function $F(u,\beta)=\Lambda^{\theta
,\beta,a}(u)$ and taking into account that $F(a/\rho,0)=0$ and,
\begin{align*}
\left.  \frac{\partial}{\partial u}F(u,\beta)\right\vert _{(u,\beta
)=(a/\rho,0)}  &  =\left.  \left(  -\rho+\rho\beta\frac{\kappa_{L}%
^{\prime\prime}(u+\theta)}{\kappa_{L}^{\prime\prime}(\theta)}\right)
\right\vert _{(u,\beta)=(a/\rho,0)}=-\rho<0,
\end{align*}
we can apply the implicit function theorem to the equation $F(u,\beta)=0$ that
ensures that there exists a neighborhood $U$ of $(a/\rho,0)$ in which we can
write $u_{a}^{0}=u_{a}^{0}(\beta),$ the root of $F(u,\beta)=0,$ as a function
of $\beta.$ Moreover, in $U,$ we have that
\begin{align*}
\frac{\partial}{\partial\beta}u_{a}^{0}(\beta)  &  =-\frac{\frac{\partial
}{\partial\beta}F(u_{a}^{0}(\beta),\beta)}{\frac{\partial}{\partial u}%
F(u_{a}^{0}(\beta),\beta)}=-\frac{\frac{\rho}{\kappa_{L}^{\prime\prime}%
(\theta)}(\theta+\kappa_{L}^{\prime}(u_{a}^{0}(\beta))-\kappa_{L}^{\prime
}(\theta))}{-\rho+\rho\beta\frac{\kappa_{L}^{\prime\prime}(u_{a}^{0}%
(\beta)+\theta)}{\kappa_{L}^{\prime\prime}(\theta)}}\\
&  =\frac{u_{a}^{0}(\beta)\int_{0}^{1}\kappa_{L}^{\prime\prime}(\theta+\lambda
u_{a}^{0}(\beta))d\lambda}{\kappa_{L}^{\prime\prime}(\theta)-\beta\kappa
_{L}^{\prime\prime}(\theta+u_{a}^{0}(\beta))},
\end{align*}
which is positive as long as
\[
\beta<\frac{\kappa_{L}^{\prime\prime}(\theta)}{\kappa_{L}^{\prime\prime
}(\theta+u_{a}^{0}(\beta))}.
\]
This yields that $u_{a}^{0}(\beta)$ is a well defined and strictly increasing
function of $\beta$ for $\beta\in\lbrack0,\beta_{m}]$, where $\beta_{m}$ is
the root of equation
\[
\Lambda^{\theta,\beta,a}(u^{m}(\theta,\beta_{m}))=0.
\]
Moreover, $u_{a}^{0}(0)=a/\rho$ and $u_{a}^{0}(\beta_{m})=u_{a}^{m}%
(\theta,\beta_{m}).$
If $\Lambda^{\theta,\beta,a}(u^{m}(\theta,\beta))<0,$ the existence of
$u^{0}(\theta,\beta)$ follows from Bolzano's Theorem and the uniqueness from
the fact that $\Lambda^{\theta,\beta,a}(u)$ is strictly decreasing in
$(0,u^{m}(\theta,\beta)).$
\end{proof}

We can now state our main result:

\begin{theorem}
\label{TheoRicattiParticularQGeneral} If $(\theta_{2},\beta_{2})\in
\mathcal{D}_{b}(1/2)$ and $(\theta_{1},\beta_{1})\in\mathbb{R}\times
\lbrack0,1)$ then $(\Psi_{0}^{\bar{\theta},\bar{\beta}}(t),\Psi_{1}%
^{\bar{\theta},\bar{\beta}}(t),\Psi_{2}^{\bar{\theta},\bar{\beta}}(t))$ are
$C^{1}([0,T];\mathbb{R})$ for any $T>0.$ Moreover,%
\[
(\Psi_{0}^{\bar{\theta},\bar{\beta}}(t),\Psi_{1}^{\bar{\theta},\bar{\beta}%
}(t),\Psi_{2}^{\bar{\theta},\bar{\beta}}(t))\longrightarrow\left(
\frac{\theta_{1}}{\alpha(1-\beta_{1})}+\int_{0}^{\infty}\{\kappa_{L}(\Psi
_{1}^{\bar{\theta},\bar{\beta}}(t)+\theta_{2})-\kappa_{L}(\theta
_{2})\}ds,0,0\right)  ,\quad t\rightarrow\infty,
\]
and%
\[
t^{-1}\log\left\Vert (\Psi_{1}^{\bar{\theta},\bar{\beta}}(t),\Psi_{2}%
^{\bar{\theta},\bar{\beta}}(t))\right\Vert \rightarrow\gamma,\quad
t\rightarrow\infty,
\]
where $\gamma=-\alpha(1-\beta_{1})$ or $\gamma=-\rho(1-\beta_{2}).$
\end{theorem}

\begin{proof}
First, recall from Remark~\ref{RemarkPropertiesSystemODEsQGeneral} that the
existence and uniqueness of the system of ODEs in Theorem
\ref{TheoRicattiAbstractQGeneral} can be reduced to establish existence and
uniqueness for the one dimensional non autonomous equation \eqref{Equ_1D_ODE}.
We have to study the time dependent vector field
\[
\tilde{\Lambda}_{1}^{\bar{\theta},\bar{\beta}}(t,u)\triangleq-\rho
u+\frac{e^{-2\alpha(1-\beta_{1})t}}{2}+\frac{\rho\beta_{2}}{\kappa_{L}%
^{\prime\prime}(\theta_{2})}(\kappa_{L}^{\prime}(u+\theta_{2})-\kappa
_{L}^{\prime}(\theta_{2})),\quad\bar{\beta}\in(0,1)^{2},\quad\bar{\theta}%
\in\bar{D}_{L}.
\]
Consider
\[
\mathcal{D}(\tilde{\Lambda}_{1}^{\bar{\theta},\bar{\beta}})\triangleq
\mathrm{int}(\{u\in\mathbb{R}:\tilde{\Lambda}_{1}^{\bar{\theta},\bar{\beta}%
}(t,u)<\infty\})=\mathrm{int}(\{u\in\mathbb{R}:\kappa_{L}^{\prime}%
(u+\theta_{2})<\infty\})=(-\infty,\Theta_{L}-\theta_{2}),
\]
and define
\[
\mathcal{D}\triangleq\mathrm{int}(%
{\displaystyle\bigcap\limits_{\bar{\beta}\in(0,1)]^{2},\bar{\theta}\in\bar
{D}_{L}}}
\mathcal{D}(\tilde{\Lambda}_{1}^{\bar{\theta},\bar{\beta}}))=(-\infty
,\Theta_{L}-\Theta_{L}/2)=(-\infty,\Theta_{L}/2).
\]
On the other hand, for $u,v\in\mathcal{D}(\tilde{\Lambda}_{1}^{\bar{\theta
},\bar{\beta}}),$ one has that
\[
\left\vert \tilde{\Lambda}_{1}^{\bar{\theta},\bar{\beta}}(t,u)-\tilde{\Lambda
}_{1}^{\bar{\theta},\bar{\beta}}(t,v)\right\vert \leq\rho\left\vert
u-v\right\vert +\frac{\rho\beta_{2}}{\kappa_{L}^{\prime\prime}(\theta_{2}%
)}\int_{0}^{\infty}\left\vert e^{uz}-e^{vz}\right\vert ze^{\theta_{2}z}%
\ell(dz),
\]
and%
\[
\int_{0}^{\infty}\left\vert e^{uz}-e^{vz}\right\vert ze^{\theta_{2}z}%
\ell(dz)\leq|u-v|\int_{0}^{\infty}e^{(u\vee v+\theta_{2})z}z^{2}\ell(dz),
\]
Moreover, note that
\[
\mathrm{int}(\{u\in\mathbb{R}:\int_{0}^{\infty}z^{2}e^{(u+\theta_{2})z}%
\ell(dz)<\infty\})=(-\infty,\Theta_{L}-\theta_{2})=\mathcal{D}(\tilde{\Lambda
}_{1}^{\bar{\theta},\bar{\beta}}).
\]
Hence, the vector field $\tilde{\Lambda}_{1}^{\bar{\theta},\bar{\beta}%
}(t,u),\bar{\theta}\in\bar{D}_{L},\bar{\beta}\in\lbrack0,1]^{2}$ is well
defined (i.e., finite) and locally Lipschitz in $\mathcal{D}(\tilde{\Lambda
}_{1}).$ Then, by the Picard-Lindel\"{o}f Theorem (see Theorem 3.1, page~18,
in Hale \cite{Ha69}) we have local existence and uniqueness for $\Psi
_{1}^{\bar{\theta},\bar{\beta}}(t)$ with $\Psi_{1}^{\bar{\theta},\bar{\beta}%
}(0)=0\in\mathcal{D}(\tilde{\Lambda}_{1}).$

Let us consider the autonomous vector field%
\[
\hat{\Lambda}_{1}^{\theta_{2},\beta_{2}}(u)\triangleq-\rho u+\frac{1}{2}%
+\frac{\rho\beta_{2}}{\kappa_{L}^{\prime\prime}(\theta_{2})}(\kappa
_{L}^{\prime}(u+\theta_{2})-\kappa_{L}^{\prime}(\theta_{2})),\quad\beta_{2}%
\in(0,1),\quad\theta_{2}\in D_{L}.
\]
Then, as $\hat{\Lambda}_{1}^{\theta_{2},\beta_{2}}(u)-\tilde{\Lambda}%
_{1}^{\bar{\theta},\bar{\beta}}(t,u)=\frac{1}{2}(1-e^{-2\alpha t})\geq0,$ for
$u\geq0,$ using a comparison theorem we have that the solution for the ODE
associated to $\tilde{\Lambda}_{1}^{\bar{\theta},\bar{\beta}}(t,u)$ and
starting at $0$ is bounded above by the corresponding solution to the ODE
associated to $\hat{\Lambda}_{1}^{\theta_{2},\beta_{2}}(u),$ which we will
denote by $\hat{\Psi}_{1}^{\theta_{2},\beta_{2}}(t).$ By Lemma
\ref{LemmaLambda(a)}, if $(\theta_{2},\beta_{2})\in\mathcal{D}_{b}(1/2)$ there
exists a unique $2/\rho<u_{1/2}^{0}(\theta_{2},\beta_{2})\leq u^{m}(\theta
_{2},\beta_{2})$ such that $\hat{\Lambda}_{1}^{\theta,\beta}(u)>0,$ for
$u\in(0,u_{1/2}^{0}(\theta_{2},\beta_{2}))$ and $\hat{\Lambda}_{1}%
^{\theta,\beta}(u_{1/2}^{0}(\theta_{2},\beta_{2}))=0.$ This yields that the
solution $\hat{\Psi}_{1}^{\theta_{2},\beta_{2}}(t)$ is defined for
$t\in\lbrack0,+\infty)$ and monotonously converges to $u_{1/2}^{0}(\theta
_{2},\beta_{2}),$ which is a stationary point of $\hat{\Lambda}_{1}%
^{\theta,\beta}(u).$ Hence, the solution $\Psi_{1}^{\bar{\theta},\bar{\beta}%
}(t)$ is bounded by $u_{1/2}^{0}(\theta_{2},\beta_{2})$ and defined for
$t\in\lbrack0,+\infty).$ To prove that actually $\Psi_{1}^{\bar{\theta}%
,\bar{\beta}}(t)$ converges to zero it is convenient to look at the $2$
dimensional system%
\begin{equation}%
\begin{array}
[c]{lcc}%
\frac{d}{dt}\Psi_{1}^{\bar{\theta},\bar{\beta}}(t)=-\rho\Psi_{1}^{\bar{\theta
},\bar{\beta}}(t)+\frac{(\Psi_{2}^{\bar{\theta},\bar{\beta}}(t))^{2}}{2}%
+\frac{\rho\beta_{2}}{\kappa_{L}^{\prime\prime}(\theta_{2})}(\kappa
_{L}^{\prime}(\Psi_{1}^{\bar{\theta},\bar{\beta}}(t)+\theta_{2})-\kappa
_{L}^{\prime}(\theta_{2})), &  & \Psi_{1}^{\bar{\theta},\bar{\beta}}(0)=0,\\
\frac{d}{dt}\Psi_{2}^{\bar{\theta},\bar{\beta}}(t)=-\alpha(1-\beta_{1}%
)\Psi_{2}^{\bar{\theta},\bar{\beta}}(t), &  & \Psi_{2}^{\bar{\theta}%
,\bar{\beta}}(0)=1,
\end{array}
\label{Equ_2D_ODE}%
\end{equation}
with the corresponding vector fields
\begin{align*}
\Lambda_{1}^{\bar{\theta},\bar{\beta}}(u_{1},u_{2})  &  =-\rho u_{1}%
+\frac{u_{2}^{2}}{2}+\frac{\rho\beta_{2}}{\kappa_{L}^{\prime\prime}(\theta
_{2})}(\kappa_{L}^{\prime}(u_{1}+\theta_{2})-\kappa_{L}^{\prime}(\theta
_{2})),\\
\Lambda_{2}^{\bar{\theta},\bar{\beta}}(u_{1},u_{2})  &  =-\alpha(1-\beta
_{1})u_{2}.
\end{align*}
Note that $(u_{1},u_{2})=(0,0)$ is a stationary point, that $\Lambda_{1}%
^{\bar{\theta},\bar{\beta}}(0,u_{2})>0$ for $u_{2}>0$, that $\Lambda_{2}%
^{\bar{\theta},\bar{\beta}}(u_{1},0)=0$ for $u_{1}>0$ and $\Lambda_{2}%
^{\bar{\theta},\bar{\beta}}(u_{1},u_{2})<0$ for $u_{1}>0,u_{2}>0.$ Hence, the
region $S_{\bar{\theta},\bar{\beta}}=\{(u_{1},u_{2}):0\leq u_{1}<\Theta
_{L}-\theta_{2},0\leq u_{2}\leq1\}$ is invariant for this vector field, i.e.,
a solution that enters $S_{\bar{\theta},\bar{\beta}}$ cannot leave
$S_{\bar{\theta},\bar{\beta}}.$ Moreover, we have that the vector field
$\Lambda_{1}^{\bar{\theta},\bar{\beta}}(u_{1},u_{2})$ evaluated at the line
$u_{2}=0$ has the form
\[
\Lambda_{1}^{\bar{\theta},\bar{\beta}}(u_{1},0)=-\rho u_{1}+\frac{\rho
\beta_{2}}{\kappa_{L}^{\prime\prime}(\theta_{2})}(\kappa_{L}^{\prime}%
(u_{1}+\theta_{2})-\kappa_{L}^{\prime}(\theta_{2})),
\]
i.e., $\Lambda_{1}^{\bar{\theta},\bar{\beta}}(u_{1},0)=\Lambda^{\theta
_{2},\beta_{2},0}(u_{1}).$ By Lemma \ref{LemmaLambda(a)}, it then follows that
$\Lambda_{1}^{\bar{\theta},\bar{\beta}}(u_{1},0)=\Lambda^{\theta_{2},\beta
_{2},0}(u_{1})<0,$ for $u_{1}\in(0,u^{m}(\theta_{2},\beta_{2})).$ In addition,
if $\left(  \theta_{2},\beta_{2}\right)  \in\mathcal{D}(1/2),$ we have that
$u_{1/2}^{0}(\theta_{2},\beta_{2})<u^{m}(\theta_{2},\beta_{2}).$ This means
that $\Lambda_{1}^{\bar{\theta},\bar{\beta}}(u_{1},0)<0$ for $u_{1}%
\in(0,u_{1/2}^{0}(\theta_{2},\beta_{2})),$ which can be extended to $\left(
u_{1},u_{2}\right)  \in(0,u_{1/2}^{0}(\theta_{2},\beta_{2}))\times
(0,\delta)\triangleq R_{\bar{\theta},\bar{\beta}}(\delta)$ for some
$0<\delta<1.$ Note that $R_{\bar{\theta},\bar{\beta}}(\delta)$ is in the
domain of attraction of the stationary point $(0,0).$ As $\Psi_{2}%
^{\bar{\theta},\bar{\beta}}(t)=e^{-\alpha(1-\beta_{1})t}$ and $\Psi_{1}%
^{\bar{\theta},\bar{\beta}}(t)<u_{1/2}^{0}(\theta_{2},\beta_{2})$,\ we have
that $(\Psi_{1}^{\bar{\theta},\bar{\beta}}(t),\Psi_{2}^{\bar{\theta}%
,\bar{\beta}}(t))\in\mathrm{int}(R_{\bar{\theta},\bar{\beta}}(\delta))$ for
$t>-\frac{\log(\delta)}{\alpha(1-\beta_{1})}$ and, hence, it converges to
$(0,0)$ when $t$ tends to infinity. Note that we can look at the system
$\left(  \ref{Equ_2D_ODE}\right)  $ as a perturbed linear system, i.e.,%
\[%
\begin{array}
[c]{lcc}%
\frac{d}{dt}\Psi_{1}^{\bar{\theta},\bar{\beta}}(t)=-\rho(1-\beta_{2})\Psi
_{1}^{\bar{\theta},\bar{\beta}}(t)+G_{1}(\Psi_{1}^{\bar{\theta},\bar{\beta}%
}(t),\Psi_{2}^{\bar{\theta},\bar{\beta}}(t)), &  & \Psi_{1}^{\bar{\theta}%
,\bar{\beta}}(0)=0,\\
\frac{d}{dt}\Psi_{2}^{\bar{\theta},\bar{\beta}}(t)=-\alpha(1-\beta_{1}%
)\Psi_{2}^{\bar{\theta},\bar{\beta}}(t)+G_{2}(\Psi_{1}^{\bar{\theta}%
,\bar{\beta}}(t),\Psi_{2}^{\bar{\theta},\bar{\beta}}(t)), &  & \Psi_{2}%
^{\bar{\theta},\bar{\beta}}(0)=1,
\end{array}
\]
where
\begin{align*}
G_{1}(u_{1},u_{2})  &  =\frac{u_{2}^{2}}{2}+\frac{\rho\beta_{2}}{\kappa
_{L}^{\prime\prime}(\theta_{2})}(\int_{0}^{1}\int_{0}^{1}\kappa_{L}%
^{\prime\prime\prime}(\theta_{2}+\lambda_{1}\lambda_{2}u_{1})d\lambda
_{2}\lambda_{1}d\lambda_{1})u_{1}^{2},\\
G_{2}(u_{1},u_{2})  &  =0,
\end{align*}
and
\[
\lim_{(u_{1},u_{2})\rightarrow(0,0)}\frac{(G_{1}(u_{1},u_{2}),G_{2}%
(u_{1},u_{2}))}{\sqrt{u_{1}^{2}+u_{2}^{2}}}=(0,0).
\]
Hence, that $(\Psi_{1}^{\bar{\theta},\bar{\beta}}(t),\Psi_{2}^{\bar{\theta
},\bar{\beta}}(t))$ converges to zero exponentially fast follows from Theorem
3.1.(i), Chapter VII, in Hartman \cite{Har64}. On the other hand, by Remark
\ref{RemarkPropertiesSystemODEsQGeneral} and the monotone convergence theorem,
we have that
\begin{align*}
\lim_{t\rightarrow\infty}\Psi_{0}^{\bar{\theta},\bar{\beta}}(t)  &
=\lim_{t\rightarrow\infty}\theta_{1}\frac{1-e^{-\alpha(1-\beta_{1})t}}%
{\alpha(1-\beta_{1})}+\int_{0}^{t}\{\kappa_{L}(\Psi_{1}^{\bar{\theta}%
,\bar{\beta}}(s)+\theta_{2})-\kappa_{L}(\theta_{2})\}ds\\
&  =\frac{\theta_{1}}{\alpha(1-\beta_{1})}+\int_{0}^{\infty}\{\kappa_{L}%
(\Psi_{1}^{\bar{\theta},\bar{\beta}}(s)+\theta_{2})-\kappa_{L}(\theta
_{2})\}ds<\infty.
\end{align*}
To prove that the previous integral is finite, note first that, as $\Psi
_{1}^{\bar{\theta},\bar{\beta}}(t)<u_{1/2}^{0}(\theta_{2},\beta_{2})$ and the
function $\kappa_{L}(u)$ is increasing we have that $\kappa_{L}(\Psi_{1}%
^{\bar{\theta},\bar{\beta}}(t)+\theta_{2})\leq\kappa_{L}(u_{1/2}^{0}%
(\theta_{2},\beta_{2})+\theta_{2}).$ But, by definition
\[
0=-\rho u_{1/2}^{0}(\theta_{2},\beta_{2})+\frac{1}{2}+\frac{\rho\beta_{2}%
}{\kappa_{L}^{\prime\prime}(\theta_{2})}(\kappa_{L}^{\prime}(u_{1/2}%
^{0}(\theta_{2},\beta_{2})+\theta_{2})-\kappa_{L}^{\prime}(\theta_{2})),
\]
which yields that
\[
\kappa_{L}^{\prime}(u_{1/2}^{0}(\theta_{2},\beta_{2})+\theta_{2})=\frac
{\kappa_{L}^{\prime\prime}(\theta_{2})}{\rho\beta_{2}}(\rho u_{1/2}^{0}%
(\theta_{2},\beta_{2})-\frac{1}{2})+\kappa_{L}^{\prime}(\theta_{2}),
\]
which is bounded. Hence, it suffices to prove that $\kappa_{L}(\Psi_{1}%
^{\bar{\theta},\bar{\beta}}(t)+\theta_{2})-\kappa_{L}(\theta_{2})$ converges
to zero faster than $t^{-(1+\varepsilon)},$ for some $\varepsilon>0,$ when $t$
tends to infinity. We have that
\begin{align*}
\lim_{t\rightarrow\infty}t^{(1+\varepsilon)}(\kappa_{L}(\Psi_{1}^{\bar{\theta
},\bar{\beta}}(t)+\theta_{2})-\kappa_{L}(\theta_{2}))  &  =\lim_{t\rightarrow
\infty}t^{(1+\varepsilon)}\Psi_{1}^{\bar{\theta},\bar{\beta}}(t)\int_{0}%
^{1}\kappa_{L}^{\prime}(\theta_{2}+\lambda\Psi_{1}^{\bar{\theta},\bar{\beta}%
}(t))d\lambda\\
&  =\left(  \lim_{t\rightarrow\infty}t^{(1+\varepsilon)}\Psi_{1}^{\bar{\theta
},\bar{\beta}}(t)\right)  \left(  \lim_{t\rightarrow\infty}\int_{0}^{1}%
\kappa_{L}^{\prime}(\theta_{2}+\lambda\Psi_{1}^{\bar{\theta},\bar{\beta}%
}(t))d\lambda\right) \\
&  =\kappa_{L}^{\prime}(\theta_{2})\lim_{t\rightarrow\infty}t^{(1+\varepsilon
)}\Psi_{1}^{\bar{\theta},\bar{\beta}}(t)=0,
\end{align*}
because $\Psi_{1}^{\bar{\theta},\bar{\beta}}(t)$ converges to zero
exponentially fast and $\lim_{t\rightarrow\infty}\int_{0}^{1}\kappa
_{L}^{\prime}(\theta_{2}+\lambda\Psi_{1}^{\bar{\theta},\bar{\beta}%
}(t))d\lambda=\kappa_{L}^{\prime}(\theta_{2})$ by bounded convergence.
\end{proof}

An immediate consequence of the Theorem above is that the forward price will
be equal to the seasonal function $\Lambda_{g}(T)$ in the long end, that is,
when $(\theta_{2},\beta_{2})\in\mathcal{D}_{b}(1/2)$ and $(\theta_{1}%
,\beta_{1})\in\mathbb{R}\times\lbrack0,1)$, it holds that
\[
\lim_{T\rightarrow\infty}\frac{F_{Q}(t,T)}{\Lambda_{g}(T)}=\exp\left(
\frac{\theta_{1}}{\alpha(1-\beta_{1})}+\int_{0}^{\infty}\{\kappa_{L}(\Psi
_{1}^{\bar{\theta},\bar{\beta}}(s)+\theta_{2})-\kappa_{L}(\theta
_{2})\}\,ds\right)  .
\]
Note that to have this limiting de-seasonalized forward price, we must compute
an integral of a nonlinear function of $\Psi_{1}^{\bar{\theta},\bar{\beta}%
}(t)$, for which we do not have any explicit solution available. Note that
from part 4(b) in Lemma~\ref{LemmaLambda(a)} we have $(\theta_{2},\beta
_{2})\in\mathcal{D}_{b}(1/2)$ if $\theta_{2}<\Theta_{L}-1/2\rho$ and
$\beta_{2}\in\lbrack0,\beta_{m}]$, for a uniquely defined $0<\beta_{m}<1$. We
recall that $\rho$ is the speed of mean reversion of the stochastic volatility
$\sigma^{2}(t)$, and $\Theta_{L}$ is the maximal exponential integrability of
$L$, the subordinator driving the same process. Thus, we must choose
$\theta_{2}$ less that $\Theta_{L}$, by a distance given by the inverse of the
speed of mean reversion. Then we know there exists an interval of $\beta_{2}%
$'s for which we can reduce the speed of mean reversion of $\sigma^{2}(t)$.
Here we see clearly the competition between the jumps of $L$ and the speed of
mean reversion of $\sigma^{2}(t)$.

We note that if we just change the levels of mean reversion, that is assuming
$\bar{\beta}=(0,0),$ then we can compute the risk premium more explicitly.
This case will correspond to an Esscher transform of both the Brownian motion
driving $X$ and the subordinator $L$ driving $\sigma^{2}(t)$.

\begin{proposition}
\label{PropRPEsscher}Suppose that $\bar{\beta}=(0,0)$ and $\bar{\theta}%
\in\mathbb{R}\times D_{L}.$ Then the forward price is given by
\begin{align*}
\mathbb{E}_{Q}[\exp(X(T))|\mathcal{F}_{t}] &  =\exp\left(  e^{-\rho(T-t)}%
\frac{1-e^{-(2\alpha-\rho)(T-t)}}{2(2\alpha-\rho)}\sigma^{2}(t)+e^{-\alpha
(T-t)}X(t)\right.  \\
&  \qquad\left.  +\theta_{1}\frac{1-e^{-\alpha(T-t)}}{\alpha}+\int_{0}%
^{T-t}\kappa_{L}\left(  e^{-\rho s}\frac{1-e^{-(2\alpha-\rho)s}}%
{2(2\alpha-\rho)}+\theta_{2}\right)  -\kappa_{L}(\theta_{2})ds\right)  ,
\end{align*}
and the risk premium by%
\begin{align}
R_{Q}^{F}(t,T) &  =\mathbb{E}_{P}[S(T)|\mathcal{F}_{t}]\left\{  \exp\left(
\int_{0}^{T-t}\kappa_{L}\left(  e^{-\rho s}\frac{1-e^{-(2\alpha-\rho)s}%
}{2(2\alpha-\rho)}+\theta_{2}\right)  -\kappa_{L}(\theta_{2})ds\right.
\right.  \label{EquRPremiumEsscher}\\
&  \qquad-\int_{0}^{T-t}\kappa_{L}\left(  e^{-\rho s}\frac{1-e^{-(2\alpha
-\rho)s}}{2(2\alpha-\rho)}\right)  ds+\left.  \left.  \theta_{1}%
\frac{1-e^{-\alpha(T-t)}}{\alpha}\right)  -1\right\}  .\nonumber
\end{align}

\end{proposition}

\begin{proof}
Note that the system of generalised Riccati equations to solve is
\[%
\begin{array}
[c]{lcc}%
\frac{d}{dt}\Psi_{1}^{\bar{\theta},0}(t)=-\rho\Psi_{1}^{\bar{\theta}%
,0}(t)+\frac{(\Psi_{2}^{\bar{\theta},0}(t))^{2}}{2} &  & \Psi_{1}^{\bar
{\theta},0}(0)=0,\\
\frac{d}{dt}\Psi_{2}^{\bar{\theta},0}(t)=-\alpha\Psi_{2}^{\bar{\theta}%
,0}(t), &  & \Psi_{2}^{\bar{\theta},0}(0)=1,\\
\frac{d}{dt}\Psi_{0}^{\bar{\theta},0}(t)=\theta_{1}\Psi_{2}^{\bar{\theta}%
,0}(t)+\kappa_{L}(\Psi_{1}^{\bar{\theta},0}(t)+\theta_{2})-\kappa_{L}%
(\theta_{2}), &  & \Psi_{0}^{\bar{\theta},0}(0)=0.
\end{array}
\]
With respect to $\Psi_{1}^{\bar{\theta},0}(t)$ and $\Psi_{2}^{\bar{\theta}%
,0}(t),$ this coincides with the one satisfied by $\Psi_{1}^{0,0}(t)=e^{-\rho
t}\frac{1-e^{-(2\alpha-\rho)t}}{2(2\alpha-\rho)}$ and $\Psi_{2}^{0,0}%
(t)=e^{-\alpha t}.$ Hence, $\Psi_{1}^{\bar{\theta},0}(t)=e^{-\rho t}%
\frac{1-e^{-(2\alpha-\rho)t}}{2(2\alpha-\rho)},\Psi_{2}^{\bar{\theta}%
,0}(t)=e^{-\alpha t}$ and we just need to integrate the equation for $\Psi
_{0}^{\bar{\theta},0}(t)$ to obtain that%
\begin{align*}
\Psi_{0}^{\bar{\theta},0}(t)  &  =\theta_{1}\int_{0}^{t}\Psi_{2}^{\bar{\theta
},0}(s)ds+\int_{0}^{t}\kappa_{L}(\Psi_{1}^{\bar{\theta},0}(t)+\theta
_{2})-\kappa_{L}(\theta_{2})\\
&  =\theta_{1}\frac{1-e^{-\alpha t}}{\alpha}+\int_{0}^{t}\kappa_{L}\left(
e^{-\rho s}\frac{1-e^{-(2\alpha-\rho)s}}{2(2\alpha-\rho)}+\theta_{2}\right)
-\kappa_{L}(\theta_{2})ds,
\end{align*}
to conclude.
\end{proof}

Next, we present two examples where we apply the previous results.
\begin{figure}
\subfloat[][Streamplot of $(\Lambda_{1}^{\bar{\theta},\bar{\beta}%
}(u_1,u_2),\Lambda_{2}^{\bar{\theta},\bar{\beta}}(u_1,u_2))$]{\includegraphics
[width=6.0cm]{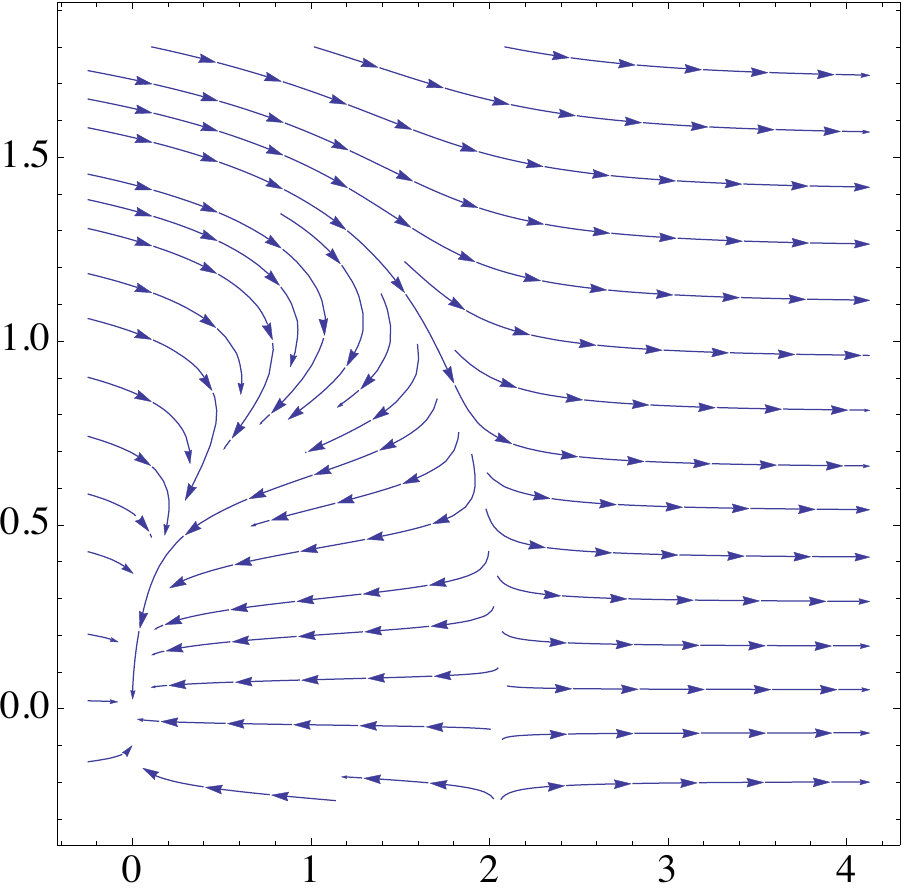}}
\qquad\subfloat[][Plot of $(\Psi_{1}^{\bar{\theta},\bar{\beta}}(t),\Psi
_{2}^{\bar{\theta},\bar{\beta}}%
(t))$ (left) and the analogous solution of the ODE with $\hat{\Lambda}%
_{1}^{\theta_2,\beta_2}(u_1)$ instead of $\Lambda_{1}^{\bar{\theta},\bar
{\beta}}(u_1,u_2)$ (right)]{\includegraphics[width=6.0cm]{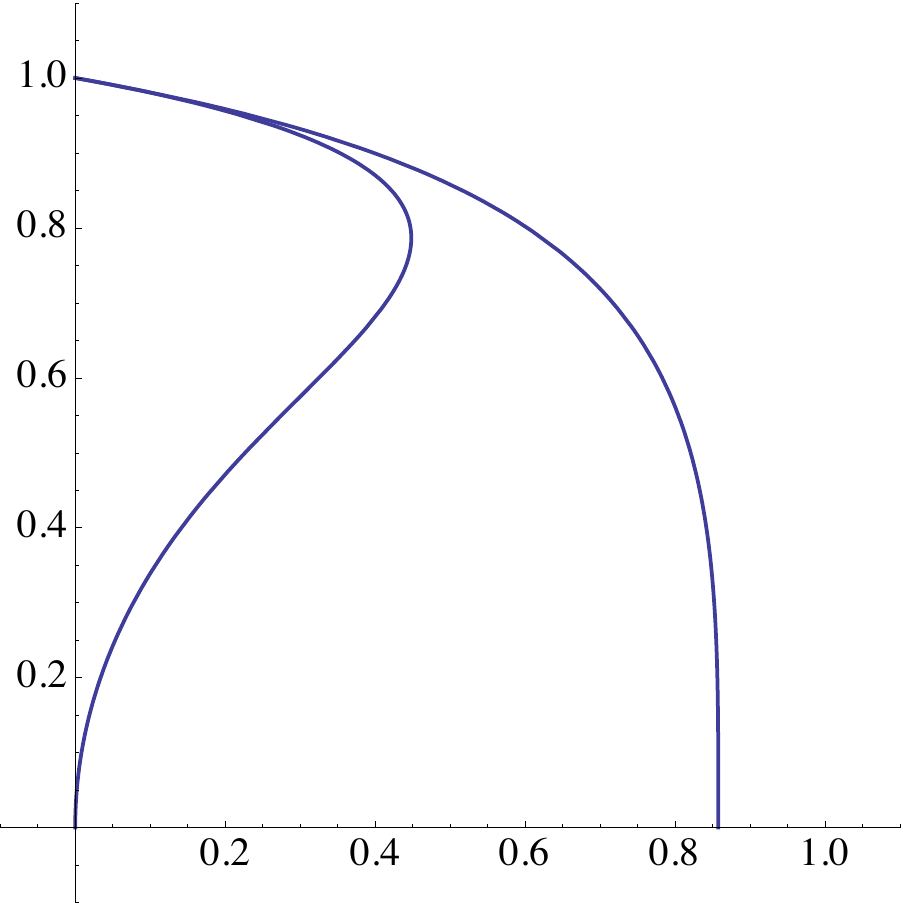}}
\\
\subfloat[][Plot of $\hat{\Lambda}_{1}^{\theta_2,\beta_2}%
(u).$ ]{\includegraphics[width=6.5cm]{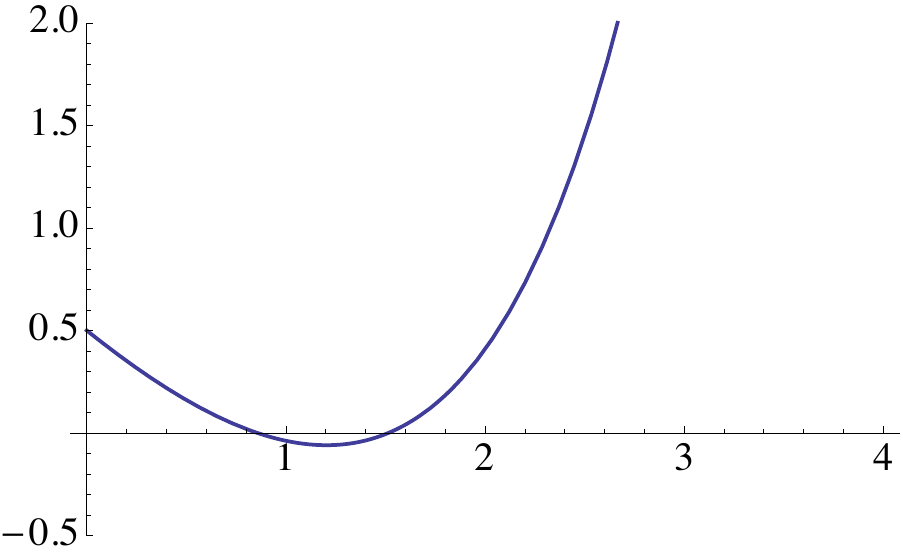}}
\qquad\subfloat[][Plot of $u_{1/2}^{m}$ (blue) and $u_{1/2}^{0}%
$ (violet).]{\includegraphics[width=6.5cm]{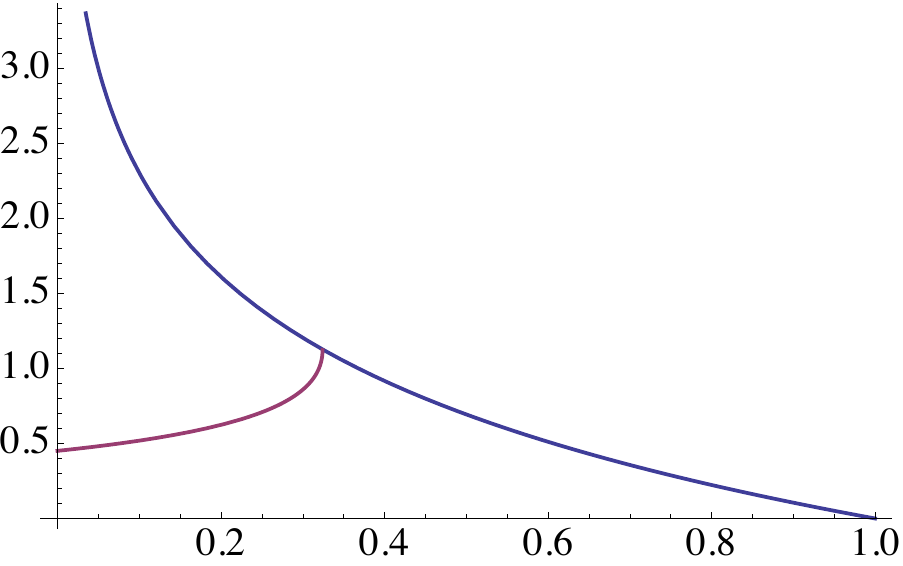}}
\caption{Some plots related to example \ref{ExampleOneJump}. We take $\rho
=1.11,\alpha=0.127,\beta_1=\beta_2=0.3.$}
\label{FigOneJump}
\end{figure}%

\begin{example}
\label{ExampleOneJump}We start by the simplest possible case. Assume that the
L\'{e}vy measure is $\delta_{\{1\}}(dz),$ that is, the L\'{e}vy process $L$
has only jumps of size $1.$ In this case $\Theta_{L}=\infty$ and, hence,
$D_{L}=\mathbb{R}.$ We have that $\kappa_{L}(\theta_{2})=e^{\theta_{2}}-1$ and
$\kappa_{L}^{(n)}(\theta_{2})=e^{\theta_{2}},n\in\mathbb{N}.$ Therefore, the
associated generalised Riccati equation is given
\begin{equation}%
\begin{array}
[c]{lcc}%
\frac{d}{dt}\Psi_{1}^{\bar{\theta},\bar{\beta}}(t)=-\rho\Psi_{1}^{\bar{\theta
},\bar{\beta}}(t)+\frac{(\Psi_{2}^{\bar{\theta},\bar{\beta}}(t))^{2}}{2}%
+\rho\beta_{2}(e^{\Psi_{1}^{\bar{\theta},\bar{\beta}}(t)}-1), &  & \Psi
_{1}^{\bar{\theta},\bar{\beta}}(0)=0,\\
\frac{d}{dt}\Psi_{2}^{\bar{\theta},\bar{\beta}}(t)=-\alpha(1-\beta_{1}%
)\Psi_{2}^{\bar{\theta},\bar{\beta}}(t), &  & \Psi_{2}^{\bar{\theta}%
,\bar{\beta}}(0)=1,\\
\frac{d}{dt}\Psi_{0}^{\bar{\theta},\bar{\beta}}(t)=\theta_{1}\Psi_{2}%
^{\bar{\theta},\bar{\beta}}(t)+e^{\theta_{2}}(e^{\Psi_{1}^{\bar{\theta}%
,\bar{\beta}}(t)}-1), &  & \Psi_{0}^{\bar{\theta},\bar{\beta}}(0)=0,
\end{array}
\label{EquationGR_delta1}%
\end{equation}
In this example%
\[
\hat{\Lambda}_{1}^{\theta_{2},\beta_{2}}(u)=\Lambda^{\theta_{2},\beta_{2}%
,1/2}(u)=\frac{1}{2}-\rho u+\rho\beta_{2}(e^{u}-1),
\]
which does not depend on $\theta_{2}.$ By Lemma \ref{LemmaLambda(a)},
$\Lambda^{\theta,\beta,1/2}(u)$ attains its minimum at
\[
u_{1/2}^{m}(\theta_{2},\beta_{2})=\log\left(  \frac{e^{\theta_{2}}}{\beta_{2}%
}\right)  -\theta_{2}=-\log\left(  \beta_{2}\right)
\]
and equation \eqref{Equ_Beta_m} reads
\begin{equation}
\Lambda^{\theta_{2},\beta_{2},1/2}(u_{1/2}^{m})=\frac{1}{2}+\rho\log\left(
\beta_{2}\right)  +\rho(1-\beta_{2})=0. \label{Equ_Beta_m_1}%
\end{equation}
Using the Lambert \textrm{W} function, i.e., the function defined by
\textrm{W}$(z)e^{\mathrm{W}(z)}=z,z\in\mathbb{C},$ we get that $\beta_{m},$
the root of equation \eqref{Equ_Beta_m_1} is given by%
\[
\beta_{m}=-\mathrm{W}(-e^{-\left(  1+\frac{1}{2\rho}\right)  }).
\]
Hence, according to Lemma \ref{LemmaLambda(a)}, the set $\mathcal{D}_{b}%
(\frac{1}{2})=\{(\theta_{2},\beta_{2}):\beta_{2}\in\lbrack0,\beta_{m}]\}$ and
if $\beta_{2}\in\lbrack0,\beta_{m}]$ there exists a unique root $u_{1/2}%
^{0}(\theta_{2},\beta_{2})$ of equation $\Lambda^{\theta_{2},\beta_{2}%
,1/2}(u)=0$ satisfying $u_{1/2}^{0}(\theta_{2},\beta_{2})\leq u_{1/2}%
^{m}(\theta_{2},\beta_{2})$. This root is given by
\[
u_{1/2}^{0}(\beta_{2})=\frac{1}{2\rho}-\left(  \beta_{2}+\mathrm{W}(-\beta
_{2}e^{(\frac{1}{2\rho}-\beta_{2})})\right)  .
\]
See Figure~\ref{FigOneJump} for a graphical illustration of this case.
\end{example}

%

\begin{figure}
\subfloat[][Streamplot of $(\Lambda_{1}^{\bar{\theta},\bar{\beta}%
}(u_1,u_2),\Lambda_{2}^{\bar{\theta},\bar{\beta}}(u_1,u_2))$]{\includegraphics
[width=6.0cm]{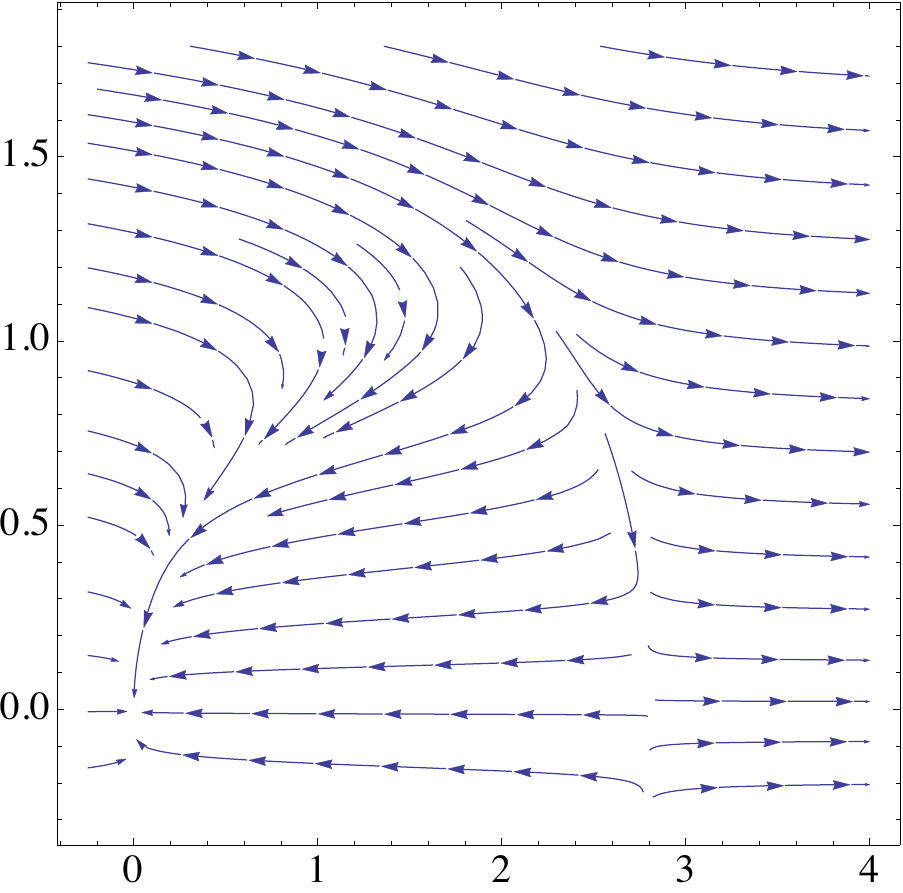}}
\qquad\subfloat[][Plot of $(\Psi_{1}^{\bar{\theta},\bar{\beta}}(t),\Psi
_{2}^{\bar{\theta},\bar{\beta}}%
(t))$ (left) and the analogous solution of the ODE with $\hat{\Lambda}%
_{1}^{\theta_2,\beta_2}(u_1)$ instead of $\Lambda_{1}^{\bar{\theta},\bar
{\beta}}(u_1,u_2)$ (right)]{\includegraphics[width=6.0cm]{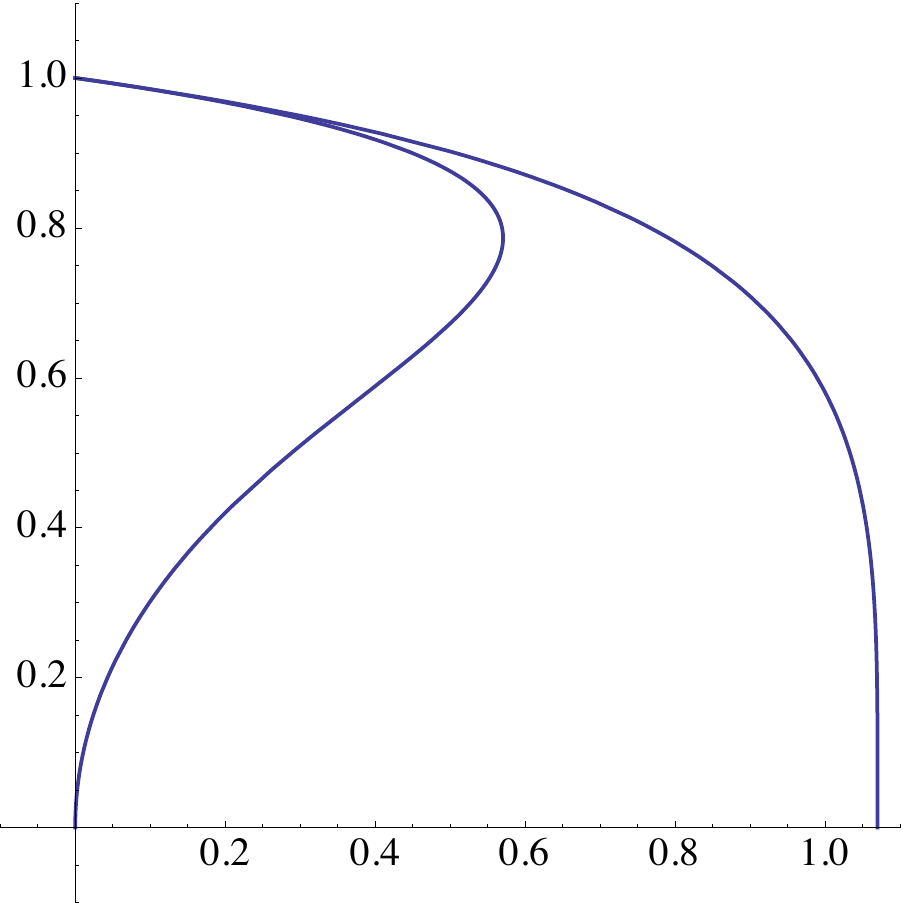}}
\\
\subfloat[][Plot of $\hat{\Lambda}_{1}^{\theta_2,\beta_2}%
(u).$ ]{\includegraphics[width=6.5cm]{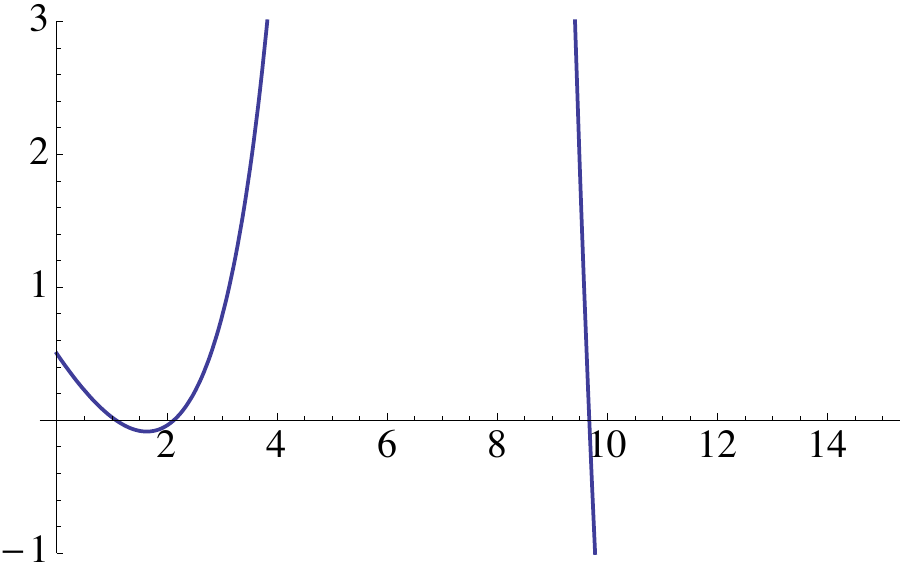}}
\qquad\subfloat[][Plot of $u_{1/2}^{m}$ (blue) and $u_{1/2}^{0}%
$ (violet).]{\includegraphics[width=6.5cm]{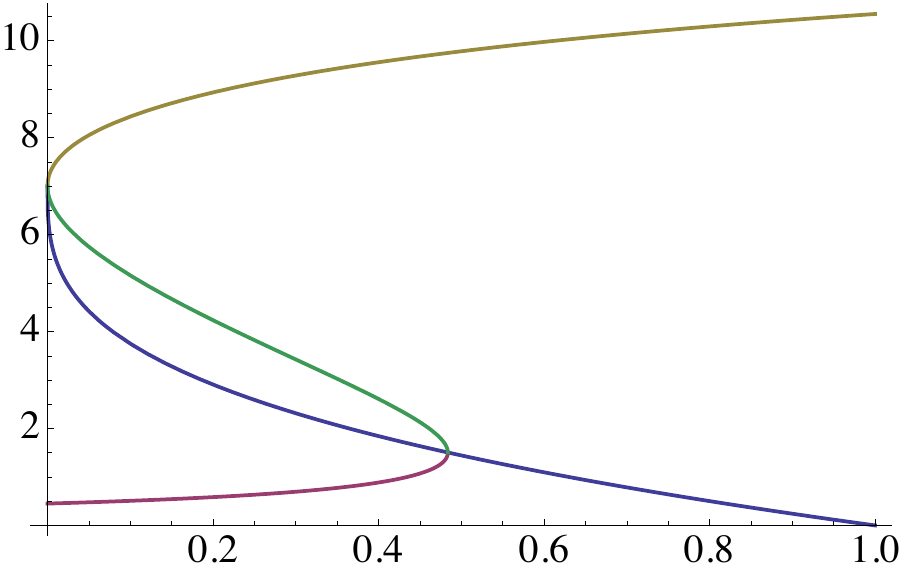}}
\caption{Some plots related to example \ref{ExampleExpJumps}. We take $\rho
=1.11,\alpha=0.127,\lambda=2,\theta_1\in\mathbb{R},\theta_2=-5,\beta
_1=\beta_2=0.45.$}
\label{FigExpJumps}
\end{figure}%

\begin{example}
\label{ExampleExpJumps}Assume that the L\'{e}vy measure is $\ell
(dz)=ce^{-\lambda z}\boldsymbol{1}_{(0,\infty)},$ that is, $L$ is a compound
Poisson process with intensity $c/\lambda$ and exponentially distributed jumps
with mean $1/\lambda.$ In this case $\Theta_{L}=\lambda$ and, hence,
$D_{L}=\mathbb{(-\infty},\lambda/2).$ We have that $\kappa_{L}(\theta
)=c\theta/\lambda(\lambda-\theta)$ and $\kappa_{L}^{(n)}(\theta)=cn!/(\lambda
-\theta)^{n+1},n\in\mathbb{N}.$ Therefore, the associated generalised Riccati
equation is given by
\begin{equation}%
\begin{array}
[c]{lcc}%
\frac{d}{dt}\Psi_{1}^{\bar{\theta},\bar{\beta}}(t)=-\rho\Psi_{1}^{\bar{\theta
},\bar{\beta}}(t)+\frac{(\Psi_{2}^{\bar{\theta},\bar{\beta}}(t))^{2}}{2}%
+\frac{\rho\beta_{2}(\lambda-\theta_{2})^{3}}{2}\left\{  \frac{1}%
{(\lambda-\theta_{2}-\Psi_{1}^{\bar{\theta},\bar{\beta}}(t))^{2}}-\frac
{1}{(\lambda-\theta_{2})^{2}}\right\}  , &  & \Psi_{1}^{\bar{\theta}%
,\bar{\beta}}(0)=0,\\
\frac{d}{dt}\Psi_{2}^{\bar{\theta},\bar{\beta}}(t)=-\alpha(1-\beta_{1}%
)\Psi_{2}^{\bar{\theta},\bar{\beta}}(t), &  & \Psi_{2}^{\bar{\theta}%
,\bar{\beta}}(0)=1,\\
\frac{d}{dt}\Psi_{0}^{\bar{\theta},\bar{\beta}}(t)=\theta_{1}\Psi_{2}%
^{\bar{\theta},\bar{\beta}}(t)+\frac{c(\Psi_{1}^{\bar{\theta},\bar{\beta}%
}(t)+\theta_{2})}{\lambda(\lambda-\theta_{2}-\Psi_{1}^{\bar{\theta},\bar
{\beta}}(t))}-\frac{c\theta_{2}}{\lambda(\lambda-\theta_{2})}, &  & \Psi
_{0}^{\bar{\theta},\bar{\beta}}(0)=0,
\end{array}
\label{EquationGR_Exp}%
\end{equation}
In this example,
\[
\hat{\Lambda}_{1}^{\theta_{2},\beta_{2}}(u)=\Lambda^{\theta_{2},\beta_{2}%
,1/2}(u)=\frac{1}{2}-\rho u+\frac{\rho\beta_{2}(\lambda-\theta_{2})^{3}}%
{2}\left\{  \frac{1}{(\lambda-\theta_{2}-u)^{2}}-\frac{1}{(\lambda-\theta
_{2})^{2}}\right\}  .
\]
By Lemma \ref{LemmaLambda(a)}, $\Lambda^{\theta,\beta,1/2}(u):[0,\lambda
-\theta_{2})\rightarrow\mathbb{R}$ attains its minimum at
\[
u_{1/2}^{m}(\theta_{2},\beta_{2})=\left(  \lambda-\theta_{2}\right)  \left(
1-\beta_{2}^{1/3}\right)
\]
and equation \eqref{Equ_Beta_m} reads
\begin{equation}
\Lambda^{\theta_{2},\beta_{2},1/2}(u_{1/2}^{m}(\theta_{2},\beta_{2}))=\frac
{1}{2}-\rho(\lambda-\theta_{2})+\frac{3}{2}\rho(\lambda-\theta_{2})\beta
_{2}^{1/3}-\frac{1}{2}\rho(\lambda-\theta_{2})\beta_{2}=0.
\label{Equ_Beta_m_Exp}%
\end{equation}
According to Lemma \ref{LemmaLambda(a)}, if $\theta_{2}>\lambda-\frac{1}%
{2\rho}$ then $\nexists\beta_{2}\in(0,1)$ such that $(\theta_{2},\beta_{2}%
)\in\mathcal{D}_{b}(\frac{1}{2}).$ If $\theta_{2}<\lambda-\frac{1}{2\rho}$
then there exists $\beta_{m}\in(0,1)$ such that $(\theta_{2},\beta_{2}%
)\in\mathcal{D}_{b}(\frac{1}{2})$ for all $\beta_{2}\in(0,\beta_{m})$ and
$\beta_{m}$ is the unique solution of equation $\left(  \ref{Equ_Beta_m_Exp}%
\right)  $ lying in $(0,1).$ Making the change of variable $z=\beta^{1/3},$
equation $\left(  \ref{Equ_Beta_m_Exp}\right)  $ is reduced to a cubic
equation and we get
\[
\beta_{m}=\left(  \left(  \frac{2}{\sqrt{a(\lambda,\rho,\theta_{2})^{2}%
-4}-a(\lambda,\rho,\theta_{2})}\right)  ^{1/3}+\left(  \frac{\sqrt
{a(\lambda,\rho,\theta_{2})^{2}-4}-a(\lambda,\rho,\theta_{2})}{2}\right)
^{1/3}\right)  ^{3},
\]
where
\[
a(\lambda,\rho,\theta_{2})\triangleq\frac{2\rho(\lambda-\theta_{2})-1}%
{\rho(\lambda-\theta_{2})}>0.
\]
Finally, if $(\theta_{2},\beta_{2})\in\mathcal{D}_{b}(\frac{1}{2})$ there
exists a unique root $u_{1/2}^{0}(\beta_{2})$ of equation $\Lambda^{\theta
_{2},\beta_{2},1/2}(u)=0$ satisfying $u_{1/2}^{0}(\theta_{2},\beta_{2})\leq
u_{1/2}^{m}(\theta_{2},\beta_{2})$. Making the change of variable
$y=\frac{\lambda-\theta_{2}}{\lambda-\theta_{2}-u},$ we can reduce the
equation $\Lambda^{\theta_{2},\beta_{2},1/2}(u)=0$ to the cubic equation%
\[
P_{3}(y)\triangleq\beta_{2}y^{3}-\left(  a(\lambda,\rho,\theta_{2})+\beta
_{2}\right)  y+2=0,
\]
which can be solved explicitely. Inverting the change of variable,\ we can get
an explicit expression for $u_{1/2}^{0}(\theta_{2},\beta_{2}).$ We refrain to
write this explicit formula because it is too lengthy.

See Figure \ref{FigExpJumps} for some graphical illustrations of this example.
\end{example}

\subsection{Discussion on the risk premium}

The next step is to analyse qualitatively the possible risk premium profiles
that can be obtained using our change of measure. In particular, we are
interested to be able to generate risk profiles with positive values in the
short end of the forward curve and negative values in the long end. In what
follows we shall make use of the Musiela parametrization $\tau=T-t$ and we
will slightly abuse the notation by denoting $R_{Q}^{F}(t,T)$ by $R_{Q}%
^{F}(t,\tau).$ We also fix the parameters of the model under the historical
measure $P,$ i.e., $\alpha$ and $\rho,$ and study the possible sign of
$R_{Q}^{F}(t,\tau)$ in terms of the change of measure parameters, i.e.,
$\bar{\beta}=(\beta_{1},\beta_{2})$ and $\bar{\theta}=(\theta_{1},\theta_{2})$
and the time to maturity $\tau.$

In contrast to the arithmetic model, the present time enters into the risk
premium not only through the stochastic components $X,$ but also through the
stochastic volatility $\sigma^{2}.$ Moreover, in the geometric model, the risk
premium will also depend on the parameters $\theta_{2}$ and $\beta_{2},$ which
change the level and speed of mean reversion for the volatility process. We
are going to study the cases $\bar{\theta}=(0,0),\bar{\beta}=(0,0)$ and the
general case separately. Moreover, in order to graphically illustrate the
discussion we plot the risk premium profiles obtained assuming that the
subordinator $L$ is a compound Poisson process with jump intensity
$c/\lambda>0$ and exponential jump sizes with mean $\lambda.$ That is, $L$
will have the L\'{e}vy measure given in Example \ref{Example_Subordinators}.
We shall measure the time to maturity $\tau$ in days and plot $R_{Q}%
^{F}(t,\tau)$ for different maturity periods. We fix the parameters of the
model under the historical measure $P$ using the same values as in the
arithmetic case, i.e.,
\[
\alpha=0.127,\rho=1.11,c=0.4,\lambda=2.
\]
Finally, in the sequel, we are going to suppose that we are under the
assumptions of Theorem \ref{TheoRicattiParticularQGeneral}, i.e., the values
$\theta_{2},\beta_{2}$ are such that $\left(  \theta_{2},\beta_{2}\right)
\in\mathcal{D}_{b}(1/2)$ and $\Psi_{0}^{\bar{\theta},\bar{\beta}},\Psi
_{1}^{\bar{\theta},\bar{\beta}}$ and $\Psi_{2}^{\bar{\theta},\bar{\beta}}$ are
globally defined and the exponential affine formula $\left(
\ref{Equ_Exp_Affine_Formula}\right)  $ holds.

The following lemma will help us in the discussion to follow.

\begin{lemma}
\label{LemmaSignRP} The sign of the risk premium $R_{Q}^{F}(t,\tau)$ is the
same as the sign of the function
\[
\Sigma(t,\tau)\triangleq\Psi_{0}^{\bar{\theta},\bar{\beta}}(\tau)-\Psi
_{0}^{0,0}(\tau)+(\Psi_{1}^{\bar{\theta},\bar{\beta}}(\tau)-\Psi_{1}%
^{0,0}(\tau))\sigma^{2}(t)+(\Psi_{2}^{\bar{\theta},\bar{\beta}}(\tau)-\Psi
_{2}^{0,0}(\tau))X(t).
\]
Moreover,%
\begin{align}
\lim_{\tau\rightarrow\infty}\Sigma(t,\tau)  &  =\frac{\theta_{1}}%
{\alpha(1-\beta_{1})}+\int_{0}^{\infty}\int_{0}^{1}\kappa_{L}^{\prime}\left(
\lambda\Psi_{1}^{\bar{\theta},\bar{\beta}}(s)+\theta_{2}\right)  d\lambda
\Psi_{1}^{\bar{\theta},\bar{\beta}}(s)ds\label{EquLimInfSigma}\\
&  \qquad-\int_{0}^{\infty}\int_{0}^{1}\kappa_{L}^{\prime}\left(  \lambda
e^{-\rho s}\frac{1-e^{-(2\alpha-\rho)s}}{2(2\alpha-\rho)}\right)  d\lambda
e^{-\rho s}\frac{1-e^{-(2\alpha-\rho)s}}{2(2\alpha-\rho)}ds,\nonumber
\end{align}
and%
\begin{equation}
\lim_{\tau\rightarrow0}\frac{d}{d\tau}\Sigma(t,\tau)=\theta_{1}+\alpha
\beta_{1}X(t). \label{EquLim0DSigma}%
\end{equation}

\end{lemma}

\begin{proof}
That the sign of $R_{Q}^{F}(t,\tau)$ is the same as the sign of $\Sigma
(t,\tau)$ is obvious from equation \eqref{EquRiskPremium} in
Theorem~\ref{TheoRicattiAbstractQGeneral}. From the expression for
$F_{P}(t,T)$ in Proposition~\ref{Prop_ForwardsGeometric}, we can deduce that
\begin{align*}
\Psi_{0}^{0,0}(\tau)  &  =\int_{0}^{\tau}\kappa_{L}\left(  e^{-\rho s}%
\frac{1-e^{-(2\alpha-\rho)s}}{2(2\alpha-\rho)}\right)  ds\\
&  =\int_{0}^{\tau}\int_{0}^{1}\kappa_{L}^{\prime}\left(  \lambda e^{-\rho
s}\frac{1-e^{-(2\alpha-\rho)s}}{2(2\alpha-\rho)}\right)  d\lambda e^{-\rho
s}\frac{1-e^{-(2\alpha-\rho)s}}{2(2\alpha-\rho)}ds,\\
\Psi_{1}^{0,0}(\tau)  &  =e^{-\rho\tau}\frac{1-e^{-(2\alpha-\rho)\tau}%
}{2(2\alpha-\rho)},\\
\Psi_{2}^{0,0}(\tau)  &  =e^{-\alpha\tau}.
\end{align*}
Furthermore, by Theorem~\ref{TheoRicattiParticularQGeneral}, one has that%
\begin{align*}
\lim_{\tau\rightarrow\infty}\Psi_{0}^{\bar{\theta},\bar{\beta}}(\tau)  &
=\frac{\theta_{1}}{\alpha(1-\beta_{1})}+\int_{0}^{\infty}\kappa_{L}(\Psi
_{1}^{\bar{\theta},\bar{\beta}}(s)+\theta_{2})-\kappa_{L}(\theta_{2})ds\\
&  =\frac{\theta_{1}}{\alpha(1-\beta_{1})}+\int_{0}^{\infty}\int_{0}^{1}%
\kappa_{L}^{\prime}(\lambda\Psi_{1}^{\bar{\theta},\bar{\beta}}(s)+\theta
_{2})d\lambda ds,\\
\lim_{\tau\rightarrow\infty}\Psi_{1}^{\bar{\theta},\bar{\beta}}(\tau)  &
=\lim_{\tau\rightarrow\infty}\Psi_{2}^{\bar{\theta},\bar{\beta}}(\tau)=0,
\end{align*}
which yields equation \eqref{EquLimInfSigma}. On the other hand, as $\Psi
_{2}^{\bar{\theta},\bar{\beta}}(\tau)\rightarrow1$ and $\Psi_{1}^{\bar{\theta
},\bar{\beta}}(\tau)\rightarrow0$ when $\tau$ tends to zero, we have%
\begin{align*}
\lim_{\tau\rightarrow0}\frac{d}{d\tau}(\Psi_{0}^{\bar{\theta},\bar{\beta}%
}(\tau)-\Psi_{0}^{0,0}(\tau))  &  =\lim_{\tau\rightarrow0}\{\Lambda_{0}%
^{\bar{\theta},\bar{\beta}}(\Psi_{1}^{\bar{\theta},\bar{\beta}}(\tau),\Psi
_{2}^{\bar{\theta},\bar{\beta}}(\tau))-\Lambda_{0}^{0,0}(\Psi_{1}^{0,0}%
(\tau),\Psi_{2}^{0,0}(\tau))\}\\
&  =\Lambda_{0}^{\bar{\theta},\bar{\beta}}(0,1)-\Lambda_{0}^{0,0}%
(0,1)=\theta_{1},\\
\lim_{\tau\rightarrow0}\frac{d}{d\tau}(\Psi_{1}^{\bar{\theta},\bar{\beta}%
}(\tau)-\Psi_{1}^{0,0}(\tau))  &  =\lim_{\tau\rightarrow0}\{\Lambda_{1}%
^{\bar{\theta},\bar{\beta}}(\Psi_{1}^{\bar{\theta},\bar{\beta}}(\tau),\Psi
_{2}^{\bar{\theta},\bar{\beta}}(\tau))-\Lambda_{1}^{0,0}(\Psi_{1}^{0,0}%
(\tau),\Psi_{2}^{0,0}(\tau))\}\\
&  =\Lambda_{1}^{\bar{\theta},\bar{\beta}}(0,1)-\Lambda_{1}^{0,0}%
(0,1)=1/2-1/2=0,\\
\lim_{\tau\rightarrow0}\frac{d}{d\tau}(\Psi_{2}^{\bar{\theta},\bar{\beta}%
}(\tau)-\Psi_{2}^{0,0}(\tau))  &  =\lim_{\tau\rightarrow0}\{\Lambda_{2}%
^{\bar{\theta},\bar{\beta}}(\Psi_{1}^{\bar{\theta},\bar{\beta}}(\tau),\Psi
_{2}^{\bar{\theta},\bar{\beta}}(\tau))-\Lambda_{2}^{0,0}(\Psi_{1}^{0,0}%
(\tau),\Psi_{2}^{0,0}(\tau))\}\\
&  =\Lambda_{2}^{\bar{\theta},\bar{\beta}}(0,1)-\Lambda_{2}^{0,0}%
(0,1)=-\alpha(1-\beta_{1})+\alpha=\alpha\beta_{1},
\end{align*}
which yields equation \eqref{EquLim0DSigma}. The proof is complete.
\end{proof}

%

\begin{figure}
\subfloat[][$\theta_1=0.024,\theta_2=-50$ ]{\includegraphics
[width=7.2cm]{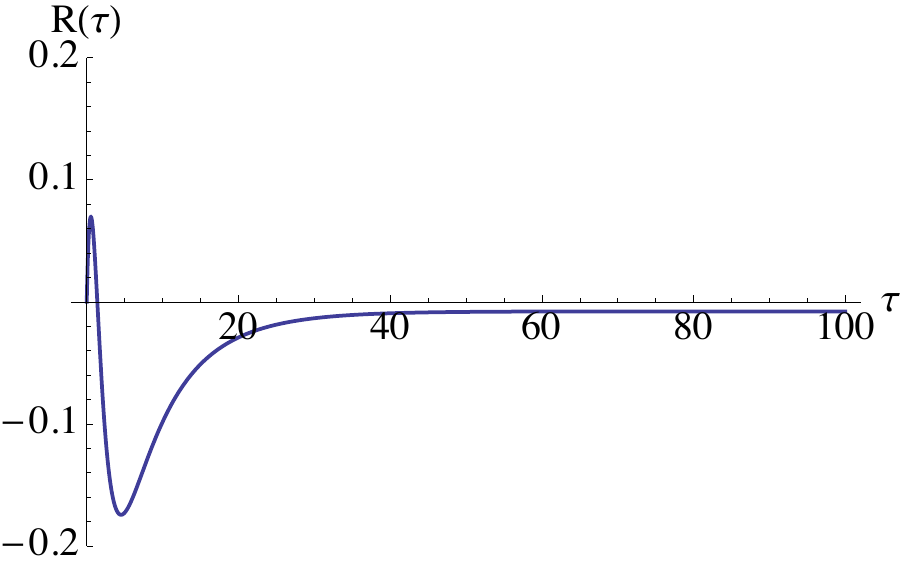}}
\qquad\subfloat[][$\theta_1=-2,\theta_2=-50$]{\includegraphics
[width=7.2cm]{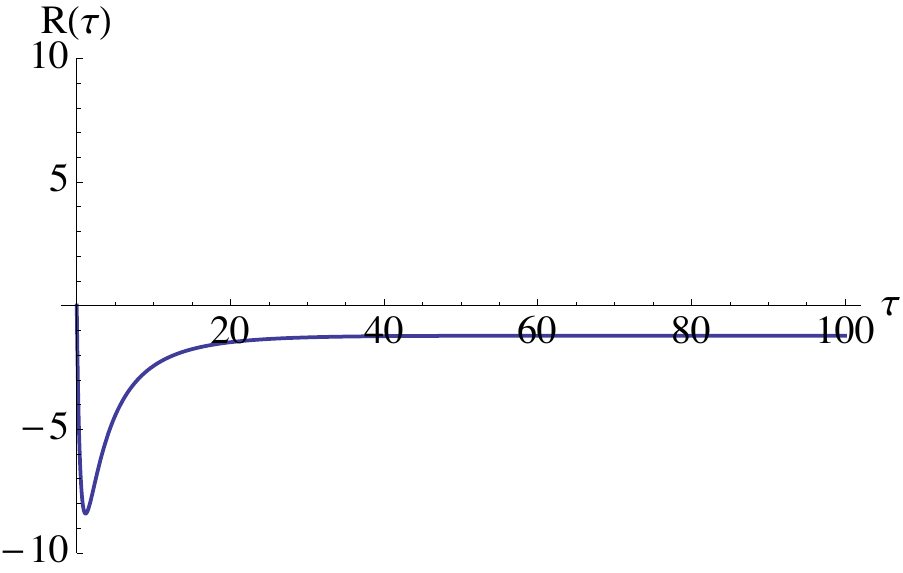}}
\caption
{Risk premium profiles when $L$ is a Compound Poisson process with exponentially distributed jumps.
We take $\rho=1.11,\alpha=0.127,\lambda=2,\c=0.4,X(t)=2.5,\sigma
(t)=0.25$  Esscher case $\beta_1=\beta_2=0.$}
\label{FigGRPEsscher}
\end{figure}%

We now continue with investigating in more detail the different cases of our
measure change.

\begin{itemize}
\item \textbf{Changing the level of mean reversion (Esscher transform)}:
Setting $\bar{\beta}=(0,0),$ the probability measure $Q$ only changes the
levels of mean reversion for the factor $X$ and the volatility process
$\sigma^{2}.$ Although the risk premium is stochastic, its sign is
deterministic. According to Proposition \ref{PropRPEsscher}, we have that the
sign of $R_{Q}^{F}(t,\tau)$ is equal to the sign of
\begin{align*}
\Sigma(t,\tau)  &  =\theta_{1}\frac{1-e^{-\alpha\tau}}{\alpha}+\int_{0}^{\tau
}\kappa_{L}\left(  e^{-\rho s}\frac{1-e^{-(2\alpha-\rho)s}}{2(2\alpha-\rho
)}+\theta_{2}\right)  -\kappa_{L}(\theta_{2})ds\\
&  \qquad-\int_{0}^{\tau}\kappa_{L}\left(  e^{-\rho s}\frac{1-e^{-(2\alpha
-\rho)s}}{2(2\alpha-\rho)}\right)  ds\\
&  =\theta_{1}\frac{1-e^{-\alpha\tau}}{\alpha}\\
&  \qquad+\theta_{2}\int_{0}^{\tau}\int_{0}^{1}\int_{0}^{1}\kappa_{L}%
^{\prime\prime}(\lambda_{2}\theta_{2}+\lambda_{1}e^{-\rho s}\frac
{1-e^{-(2\alpha-\rho)s}}{2(2\alpha-\rho)})d\lambda_{2}d\lambda_{1}e^{-\rho
s}\frac{1-e^{-(2\alpha-\rho)s}}{2(2\alpha-\rho)}ds.
\end{align*}
Moreover, by Lemma~\ref{LemmaSignRP}, equations
\eqref{EquLimInfSigma}-\eqref{EquLim0DSigma}, and the fact that $\Psi
_{1}^{\bar{\theta},0}(\tau)=e^{-\rho\tau}\frac{1-e^{-(2\alpha-\rho)\tau}%
}{2(2\alpha-\rho)}$ we get%
\[
\lim_{\tau\rightarrow\infty}\Sigma(t,\tau)=\frac{\theta_{1}}{\alpha}%
+\theta_{2}\int_{0}^{\infty}\int_{0}^{1}\int_{0}^{1}\kappa_{L}^{\prime\prime
}(\lambda_{2}\theta_{2}+\lambda_{1}\Psi_{1}^{\bar{\theta},0}(s))d\lambda
_{2}d\lambda_{1}\Psi_{1}^{\bar{\theta},0}(s)ds,
\]
and
\[
\lim_{\tau\rightarrow0}\frac{d}{d\tau}\Sigma(t,\tau)=\theta_{1}.
\]
Note that we can write
\begin{align*}
&  \theta_{2}\int_{0}^{\infty}\int_{0}^{1}\int_{0}^{1}\kappa_{L}^{\prime
\prime}(\lambda_{2}\theta_{2}+\lambda_{1}\Psi_{1}^{\bar{\theta},0}%
(s))d\lambda_{2}d\lambda_{1}\Psi_{1}^{\bar{\theta},0}(s)ds\\
&  =\int_{0}^{\infty}\int_{0}^{1}\int_{0}^{1}\theta_{2}\Psi_{1}^{\bar{\theta
},0}(s)\int_{0}^{\infty}z^{2}e^{(\lambda_{2}\theta_{2}+\lambda_{1}\Psi
_{1}^{\bar{\theta},0}(s))z}\ell(dz)d\lambda_{2}d\lambda_{1}ds\\
&  =\int_{0}^{\infty}\int_{0}^{\infty}\left(  e^{\theta_{2}z}-1\right)
(e^{\Psi_{1}^{\bar{\theta},0}(s)z}-1)\ell(dz)ds,
\end{align*}
and that $e^{\Psi_{1}^{\bar{\theta},0}(s)z}-1>0$ for $s,z>0,$ $\Psi_{1}%
^{\bar{\theta},0}(s)$ is strictly positive. Hence, if we choose $0<\theta
_{1}\ $small enough and $\theta_{2}<0$ large enough , we obtain a risk premium
which is positive in the short end of the forward curve, and negative in the
long end. Note that $\theta_{2}$ must be chosen negative.
Figure~\ref{FigGRPEsscher} shows graphically two possible risk premium curves
for given parameters as an illustration. We recall from Benth and
Ortiz-Latorre~\cite{BeO-L13} that for a two-factor mean reverting stochastic
dynamics of the spot price without stochastic volatility, we obtain similar
deterministic risk premium profiles.%

\begin{figure}
\subfloat[][$\beta_1=0.18,\beta_2=0.2$ ]{\includegraphics
[width=7.2cm]{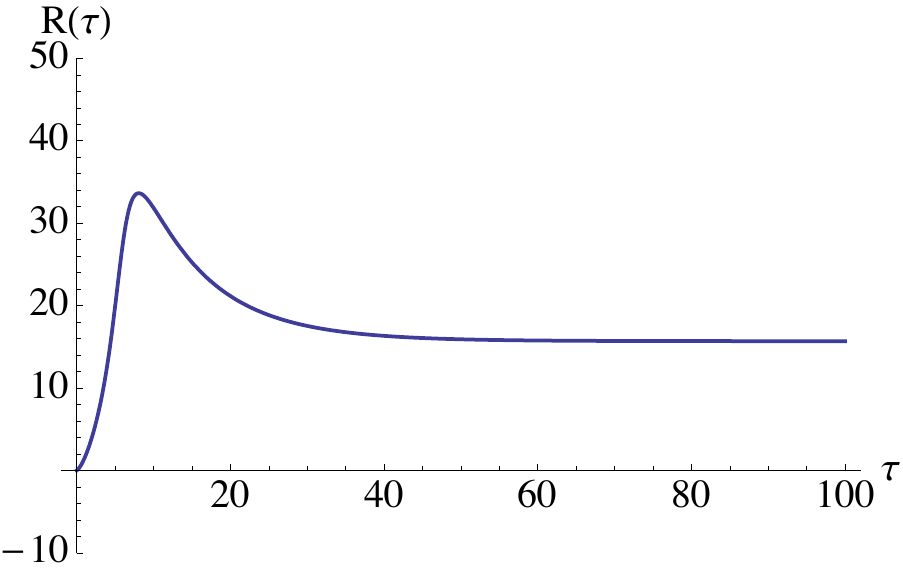}}
\qquad\subfloat[][$\beta_1=0.75,\beta_2=0$]{\includegraphics
[width=7.2cm]{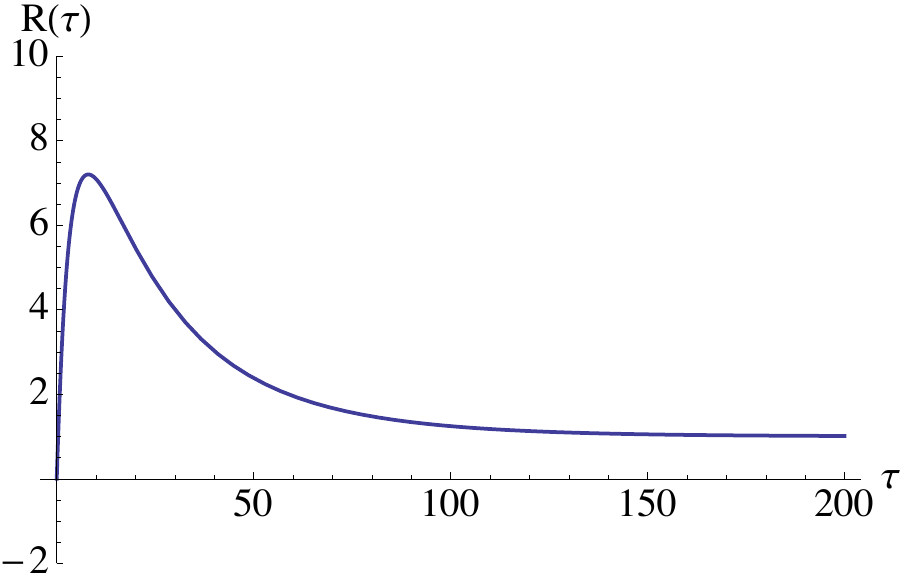}}
\caption
{Risk premium profiles when $L$ is a Compound Poisson process with exponentially distributed jumps.
We take $\rho=1.11,\alpha=0.127,\lambda=2,\c=0.4,X(t)=2.5,\sigma
(t)=0.25.$  Case $\theta_1=\theta_2=0.$}
\label{FigGRPSpeedOnly}
\end{figure}%

\item \textbf{Changing the speed of mean reversion: }Setting $\bar{\theta
}=(0,0),$ the probability measure $Q$ only changes the levels of mean
reversion for the factor $X$ and the volatility process $\sigma^{2}.$ Both the
risk premium and its sign are stochastic. According to Lemma \ref{LemmaSignRP}%
, we have that the sign of $R_{Q}^{F}(t,\tau)$ is equal to the sign of%
\begin{align*}
\Sigma(t,\tau)  &  \triangleq\Psi_{0}^{0,\bar{\beta}}(\tau)-\int_{0}^{\tau
}\int_{0}^{1}\kappa_{L}^{\prime}\left(  \lambda e^{-\rho s}\frac
{1-e^{-(2\alpha-\rho)s}}{2(2\alpha-\rho)}\right)  d\lambda e^{-\rho s}%
\frac{1-e^{-(2\alpha-\rho)s}}{2(2\alpha-\rho)}ds\\
&  \qquad+(\Psi_{1}^{0,\bar{\beta}}(\tau)-e^{-\rho\tau}\frac{1-e^{-(2\alpha
-\rho)\tau}}{2(2\alpha-\rho)})\sigma^{2}(t)+(e^{-\alpha(1-\beta_{1})\tau
}-e^{-\alpha\tau})X(\tau).
\end{align*}
Moreover,%
\begin{align*}
\lim_{\tau\rightarrow\infty}\Sigma(t,\tau)  &  =\int_{0}^{\infty}\int_{0}%
^{1}\kappa_{L}^{\prime}(\lambda\Psi_{1}^{0,\bar{\beta}}(s))d\lambda\Psi
_{1}^{0,\bar{\beta}}(s)ds\\
&  \qquad-\int_{0}^{\infty}\int_{0}^{1}\kappa_{L}^{\prime}\left(  \lambda
e^{-\rho s}\frac{1-e^{-(2\alpha-\rho)s}}{2(2\alpha-\rho)}\right)  d\lambda
e^{-\rho s}\frac{1-e^{-(2\alpha-\rho)s}}{2(2\alpha-\rho)}ds,
\end{align*}
and%
\[
\lim_{\tau\rightarrow0}\frac{d}{d\tau}\Sigma(t,\tau)=\alpha\beta_{1}X(t).
\]
Note that
\[
\Lambda^{0,\bar{\beta}}(u_{1},u_{2})-\Lambda^{0,0}(u_{1},u_{2})=\frac
{\rho\beta_{2}}{\kappa_{L}^{\prime\prime}(0)}\left(  \kappa_{L}^{\prime}%
(u_{1})-\kappa_{L}^{\prime}(0)\right)  \geq0,\quad u_{1}\geq0,
\]
and using a comparison theorem for ODEs, Theorem 6.1, page 31, in Hale
\cite{Ha69}, we get that $\Psi_{1}^{0,\bar{\beta}}(t)\geq e^{-\rho t}%
\frac{1-e^{-(2\alpha-\rho)t}}{2(2\alpha-\rho)},t\geq0.$ Hence, the risk
premium will approach to a non negative value in the long end of the forward
curve. In the short end, it can be positive or negative and stochastically
varying with $X(t).$ In Figure~\ref{FigGRPSpeedOnly} we show two different
risk premium curves, where we in particular notice the different convexity
behaviour in the short end. As all the risk premia curves will be positive in
the long end, it is not very realistic from the practical point of view to
have $\bar{\theta}=(0,0).$%

\begin{figure}
\subfloat[][$\theta_1=0.001,\theta_2=-50.0,\beta_1=0.0,\beta
_2=0.9$ ]{\includegraphics
[width=7.2cm]{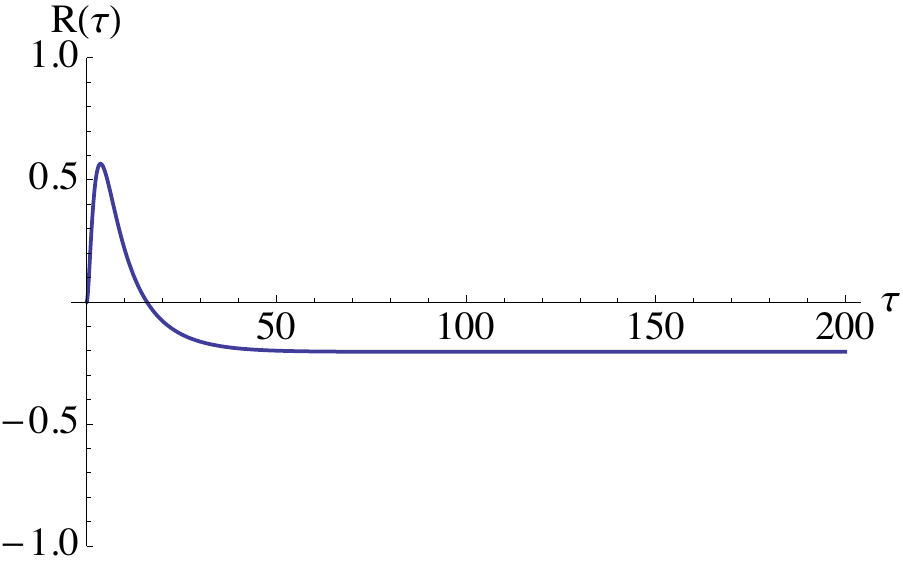}}
\qquad\subfloat[][$\theta_1=-0.1,\theta_2=-50.0,\beta_1=0.8,\beta
_2=0.8$]{\includegraphics
[width=7.2cm]{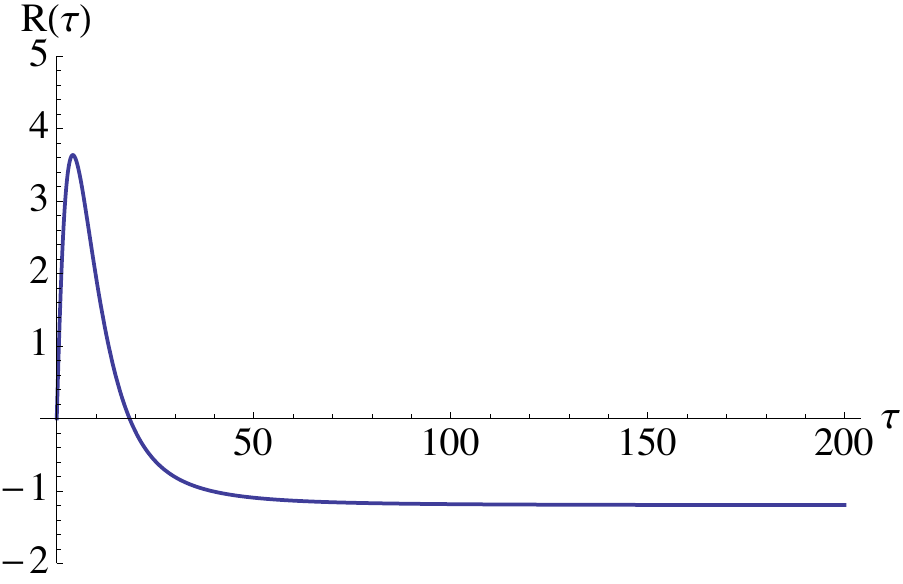}}
\caption
{Risk premium profiles when $L$ is a Compound Poisson process with exponentially distributed jumps.
We take $\rho=1.11,\alpha=0.127,\lambda=2,\c=0.4,X(t)=2.5,\sigma(t)=0.25.$ }
\label{FigGRP}
\end{figure}%

\item \textbf{Changing the level and speed of mean reversion simultaneously}:
In the general case we modify the speed and level of mean reversion for the
factor $X$ and the volatility process $\sigma^{2}$ simultaneously. According
to Lemma \ref{LemmaSignRP}, we have that
\begin{align*}
\lim_{\tau\rightarrow\infty}\Sigma(t,\tau)  &  =\frac{\theta_{1}}%
{\alpha(1-\beta_{1})}+\int_{0}^{\infty}\int_{0}^{1}\kappa_{L}^{\prime}\left(
\lambda\Psi_{1}^{\bar{\theta},\bar{\beta}}(s)+\theta_{2}\right)  d\lambda
\Psi_{1}^{\bar{\theta},\bar{\beta}}(s)ds\\
&  \qquad-\int_{0}^{\infty}\int_{0}^{1}\kappa_{L}^{\prime}\left(  \lambda
e^{-\rho s}\frac{1-e^{-(2\alpha-\rho)s}}{2(2\alpha-\rho)}\right)  d\lambda
e^{-\rho s}\frac{1-e^{-(2\alpha-\rho)s}}{2(2\alpha-\rho)}ds,
\end{align*}
and%
\[
\lim_{\tau\rightarrow0}\frac{d}{d\tau}\Sigma(t,\tau)=\theta_{1}+\alpha
\beta_{1}X(t).
\]
If we choose $\beta_{1}=0,$ then we need to prove that for some $(\theta
_{2},\beta_{2})\in\mathcal{D}_{b}\left(  1/2\right)  $ and $0<\theta_{1}$ we
have%
\begin{align}
&  \int_{0}^{\infty}\int_{0}^{1}\kappa_{L}^{\prime}\left(  \lambda e^{-\rho
s}\frac{1-e^{-(2\alpha-\rho)s}}{2(2\alpha-\rho)}\right)  d\lambda e^{-\rho
s}\frac{1-e^{-(2\alpha-\rho)s}}{2(2\alpha-\rho)}%
ds\label{EquIneqGeneralChanging}\\
&  \qquad>\frac{\theta_{1}}{\alpha}+\int_{0}^{\infty}\int_{0}^{1}\kappa
_{L}^{\prime}\left(  \lambda\Psi_{1}^{\bar{\theta},\bar{\beta}}(s)+\theta
_{2}\right)  d\lambda\Psi_{1}^{\bar{\theta},\bar{\beta}}(s)ds,\nonumber
\end{align}
in order to ensure a risk premium that changes sign from positive to negative.
In fact, inequality \ref{EquIneqGeneralChanging} holds by choosing $\theta
_{1}$ small enough and $\theta_{2}$ a large negative number, because
$\lim_{\theta_{2}\rightarrow-\infty}\kappa_{L}^{\prime}(\theta_{2})=0$. See
Figure \ref{FigGRP} for two cases.
\end{itemize}

\begin{remark}
In contrast to the arithmetic case, one can get a positive risk premium for
short time to maturity that rapidly changes to negative by just changing the
parameters of the Esscher transform, see Figure \ref{FigGRPEsscher}. Similarly
to the arithmetic case, it is not possible to get the sign change by just
modifying the speed of mean reversion of the factors, see Figure
\ref{FigGRPSpeedOnly}.\bigskip
\end{remark}

\end{document}